\documentclass[dvipsnames,a4paper,11pt,reqno]{amsart}

\usepackage{layout}
\usepackage[T1]{fontenc}
\usepackage[english]{babel}
\usepackage[utf8]{inputenc}
\usepackage{amsmath,amssymb,amsthm,amsfonts,amstext,amsopn,amsxtra,dsfont,mathrsfs,esint,enumitem,bookmark,mathtools,yhmath}
\usepackage{bbm}
\usepackage[capitalise]{cleveref}
\hypersetup{linktocpage=true}
\usepackage{hyphenat}
\usepackage{microtype}
\usepackage[headings]{fullpage}

\theoremstyle{plain}
\newtheorem{theorem}{Theorem}
\newtheorem{lemma}[theorem]{Lemma}
\newtheorem{proposition}[theorem]{Proposition}
\newtheorem{corollary}[theorem]{Corollary}

\theoremstyle{definition}

\newcommand\xqed[1]{%
	\leavevmode\unskip\penalty9999 \hbox{}\nobreak\hfill\quad\hbox{#1}%
}
\newcommand\remarkend{\xqed{$\triangle$}}

\makeatletter
\def\@endtheorem{\remarkend\endtrivlist\@endpefalse }
\makeatother
\theoremstyle{remark}

\makeatletter
\def\@endtheorem{\endtrivlist\@endpefalse }
\makeatother

\crefname{theorem}{Theorem}{Theorems}
\crefname{lemma}{Lemma}{Lemmas}
\crefname{proposition}{Proposition}{Propositions}
\crefname{corollary}{Corollary}{Corollaries}
\crefname{definition}{Definition}{Definitions}
\crefname{assumption}{Assumption}{Assumptions}
\crefname{remark}{Remark}{Remarks}
\crefname{conjecture}{Conjecture}{Conjectures}

\crefname{subsection}{subsection}{subsections}
\crefname{subsubsection}{subsection}{subsections}

\newcommand{\Tr}{\operatorname{Tr}}

\newcommand{\sgn}{\operatorname{sgn}}

\newcommand{\cI}{\mathcal{I}}
\newcommand{\cJ}{\mathcal{J}}

\newcommand{\N}{\mathbb{N}}

\newcommand{\C}{\mathbb{C}}

\newcommand{\bfA}{\mathbf{A}}
\newcommand{\bfB}{\mathbf{B}}
\newcommand{\bfD}{\mathbf{D}}

\newcommand{\bfalpha}{\boldsymbol{\alpha}}
\newcommand{\bfbeta}{\boldsymbol{\beta}}
\newcommand{\bfdelta}{\boldsymbol{\delta}}
\newcommand{\bfeps}{\boldsymbol{\varepsilon}}
\newcommand{\bfeta}{\boldsymbol{\eta}}

\newcommand{\mfH}{\mathfrak{H}}

\newcommand{\HS}{\textmd{\normalfont HS}}

\newcommand{\ceqq}{\coloneqq}

\newcommand\mydots{\ifmmode\mathellipsis\else.\kern-0.08em.\kern-0.08em.\fi}

\allowdisplaybreaks[4]

\DeclareUnicodeCharacter{2217}{*}

\usepackage{csquotes}
\usepackage[backend=bibtex,sorting=nty,citestyle=numeric-comp,bibstyle=ieee,maxbibnames=99,doi=false,isbn=false,url=false,eprint=false,dashed=true]{biblatex}
\AtEveryBibitem{\clearfield{month}}
\AtEveryBibitem{\clearfield{day}}
\usepackage{graphicx}
\usepackage{caption}
\usepackage{subcaption}

\usepackage[table,xcdraw]{xcolor}

\usepackage{tikz}
\usetikzlibrary{arrows}
\usepackage{tkz-tab}

\usepackage{standalone}

\usepackage{pgfplots}
\pgfplotsset{compat=1.15}
\usepackage{mathrsfs}
\usetikzlibrary{arrows}

\usepackage{hyperref}
\hypersetup{bookmarksdepth=3}
\title[Hilbert--Schmidt norm of fermionic reduced density matrices]{Hilbert--Schmidt norm estimates for fermionic reduced density matrices}
\author{François L. A. Visconti}
\address{Department of Mathematics, LMU Munich, Theresienstrasse 39, 80333 Munich, Germany}
\email{visconti@math.lmu.de}

\begin{document}
	\maketitle
	
	\begin{abstract}
		We prove that the Hilbert--Schmidt norm of $k$-particle reduced density matrices of $N$-body fermionic states is bounded by $C_kN^{k/2}$ - matching the scaling behaviour of Slater determinant states. This generalises a recent result of Christiansen \cite{Christiansen2024HSEstimates} on $2$-particle reduced density matrices to higher-order density matrices. Moreover, our estimate directly yields a lower bound on the von Neumann entropy and the 2-Rényi entropy of reduced density matrices, thereby providing further insight into conjectures of Carlen--Lieb--Reuvers \cite{Carlen2016EntropyEnanglement,Reuvers2018AlgorithmEE}.
	\end{abstract}
	
	\tableofcontents
	
	\section{Introduction}
	Let $(\mathfrak{H},\langle\cdot,\cdot\rangle)$ be a separable Hilbert space. We consider the $N$-body fermionic Hilbert space $\mathfrak{H}^{\wedge N}$ consisting of wavefunctions $\Psi\in\mathfrak{H}^N$ which are antisymmetric with respect to exchange of variables, meaning that they satisfy 
\begin{equation}
	\label{eq:wavefunction_antisymmetry}
	U_\sigma\Psi = \sgn(\sigma)\Psi,
\end{equation}
for all permutations $\sigma$ of $\{1,\dots,N\}$. Here $\sgn(\sigma)$ denotes the sign of $\sigma$ and $U_\sigma$ is the permutation operator defined by
\begin{equation}
	\label{eq:permutation_operator}
	U_\sigma u_1\otimes\dots\otimes u_N \coloneqq u_{\sigma(1)}\otimes\dots\otimes u_{\sigma(N)},
\end{equation}
for all $u_1,\dots,u_N\in\mathfrak{H}$.

Given a normalised state $\Psi\in\mathfrak{H}^{\wedge N}$, we define the $k$-particle reduced density matrix as
\begin{equation}
	\label{eq:k_prdm_def}
	\Gamma^{(k)} \coloneqq {N \choose k}\Tr_{k+1\rightarrow N}\left|\Psi\right\rangle\left\langle\Psi\right|.
\end{equation}
It is well-known that $\Gamma^{(k)}$ is nonnegative and trace-class \cite{Simon1979TraceIdeals} with
\begin{equation}
	\label{eq:k_prdm_trace}
	\Tr\Gamma^{(k)} = {N \choose k}.
\end{equation}
Therefore, we have the trivial bound $\Vert\Gamma^{(k)}\Vert_\textmd{op} \leq {N \choose k}$. Though this estimate is optimal in the \textit{bosonic} case, it is far from it in the fermionic one. Indeed, for $\Gamma^{(1)}$, the well-known \textit{Pauli exclusion principle} implies the much stronger bound $\Vert\Gamma^{(1)}\Vert_\textmd{op} \leq 1$, which is optimised by \textit{Slater determinants}, and for $\Gamma^{(2)}$, Yang \cite{Yang1962ConceptOR} proved the optimal bound $\Vert\Gamma^{(2)}\Vert_\textmd{op} \leq N$, which is remarkably not maximised by Slater determinants. \footnote{The optimisers are \textit{Yang pairing states}, which play a role in the BCS theory of superconductivity.} More generally, Yang \cite{Yang1963PropRDM} ($k$ even) and Bell \cite{Bell1962conjectureCNYang} ($k$ odd) proved the bound $\Vert\Gamma^{(k)}\Vert_\textmd{op} \leq C_kN^{\lfloor k/2\rfloor}$. Though the constant is not optimal, the bound can be shown to be of the right order using a trial state argument (see \cite{Carlen2016EntropyEnanglement,Reuvers2018AlgorithmEE} for a conjecture on the optimal constant).

It is easy to see, using for example Coleman's theorem \cites{Coleman1963SFDM}[Th. 3.2]{Lieb2010Stability} or the Schmidt decomposition, that the Hilbert--Schmidt norm of $\Gamma^{(1)}$ is maximised by Slater determinants, meaning that it obeys $\Vert\Gamma^{(1)}\Vert_\textmd{HS} \leq N^{1/2}$. More generally, thanks to the estimate $\Vert\Gamma^{(k)}\Vert_\textmd{op} \lesssim N^{\lfloor k/2\rfloor}$ and the identity \eqref{eq:k_prdm_trace}, we can directly deduce that the Hilbert--Schmidt norm of $\Gamma^{(k)}$ is bounded by
\begin{equation}
	\label{eq:k_prdm_hs_norm_naive_bound}
	\Vert\Gamma^{(k)}\Vert_{\textmd{HS}} \leq \sqrt{\left\Vert\Gamma^{(k)}\right\Vert_\textmd{op}\Tr\Gamma^{(k)}} \lesssim N^{k/2}N^{\lfloor k/2\rfloor/2}.
\end{equation}
Considering that we used an identity and an almost optimal bound, one might be tempted to think that \eqref{eq:k_prdm_hs_norm_naive_bound} is almost optimal as well. This is however not the case at all for $k\geq 2$. Indeed, the case $k \geq 2$ is an open problem for which Carlen--Lieb--Reuvers \cite[Conjecture 2.5]{Carlen2016EntropyEnanglement} ($k = 2$)\cite[Conjecture 5.10]{Reuvers2018AlgorithmEE} ($k\geq2$) conjectured that the Hilbert--Schmidt norm of $\Gamma^{(k)}$ is maximised by Slater determinants, that is satisfying
\begin{equation}
	\label{eq:k_prdm_hs_estimate_optimal_conjecture}
	\Vert\Gamma^{(k)}\Vert_{\textmd{HS}} \leq {N \choose k}^{1/2}.
\end{equation}
Their conjecture was motivated by the weaker conjecture \cite[Conjecture 2.4]{Carlen2016EntropyEnanglement},\cite[Conjecture 5.10]{Reuvers2018AlgorithmEE} that the von Neumann entropy of $k$-particle reduced density matrices is minimised by Slater determinants, meaning that
\begin{equation}
	\label{eq:von_neumann_entropy_lower_bound}
	S(\gamma^{(k)}) \geq \log{N \choose k},
\end{equation}
where $S$ denotes the von Neumann entropy \eqref{eq:von_neumann_entropy} and $\gamma^{(k)}$ is the \textit{trace normalised} k-particle reduced density matrix of $\Psi$. The best-known result in this direction is the nearly optimal bound
\begin{equation}
	\label{eq:hs_norm_estimate_nearly_optimal}
	\Vert\Gamma^{(2)}\Vert_\textmd{HS} \leq \sqrt{5}N/2,
\end{equation}
which was proven recently by Christiansen \cite{Christiansen2024HSEstimates}. Moreover, the weaker bound
\begin{equation}
	\label{eq:von_neumann_entropy_lower_bound_non_optimal}
	S\left(\gamma^{(k)}\right) \geq 2\log N + \mathcal{O}(1)
\end{equation}
as $N \rightarrow \infty$ has also been proven recently by Christiansen \cite{Christiansen2024HSEstimates} for $k = 2$ (in accordance with \cite[Conjecture 2.6]{Carlen2016EntropyEnanglement}) and generalised to any $k\geq 2$ by Lemm \cite{Lemm2024ERI}. Note that the case $k = 2$ had already been proven by Lemm \cite{Lemm2017EFRD} under the much more restrictive assumption that the Hilbert space $\mathfrak{H}$ has finite dimension $d \geq N$ not too far from $N$. Though the estimate \eqref{eq:von_neumann_entropy_lower_bound_non_optimal} is of the correct order for $k = 2$, it is off by a factor $k/2$ for $k \geq 3$.

The goal of the present paper is to generalised Christiansen's bound \eqref{eq:hs_norm_estimate_nearly_optimal} to higher-order reduced density matrices. As a consequence, we obtain \eqref{eq:von_neumann_entropy_lower_bound_non_optimal} with the correct factor in front of $\log N$.

\begin{theorem}
	\label{th:HS_estimates}
	Let $\Psi \in \mathfrak{H}^{\wedge N}$ be normalised and define its $k$-particle reduced density matrix $\Gamma^{(k)}$ as in \eqref{eq:k_prdm_def}. Then,
	\begin{equation}
		\label{eq:k_prdm_hs_estimate}
		\Vert\Gamma^{(k)}\Vert_{\textmd{\normalfont HS}} \leq C_kN^{k/2},
	\end{equation}
	for some constant $C_k$ that depends only on $k$.
\end{theorem}
While the constant in \eqref{eq:k_prdm_hs_estimate} is not optimal, the estimate  exhibits the correct asymptotic order, consistent with the conjectured bound \eqref{eq:k_prdm_hs_estimate_optimal_conjecture}. This result is particularly interesting when put into perspective with the bound $\Vert\Gamma^{(k)}\Vert_\textmd{op} \lesssim N^{\lfloor k/2\rfloor}$ and the normalisation condition \eqref{eq:k_prdm_trace}. Note first that \eqref{eq:k_prdm_hs_estimate} does not display the peculiar dependency in the parity of $k$ that the operator-norm bound does. What this roughly says is that in the case where $k$ is even, $\Gamma^{(k)}$ can have large eigenvalues of order $N^{k/2}$ but it cannot have too many of them, whereas in the odd case $\Gamma^{(k)}$ cannot even have large eigenvalues of order $N^{k/2}$.

The bound \eqref{eq:k_prdm_hs_estimate} has immediate consequences on the entanglement entropy of the \textit{trace normalised} $k$-particle reduced density matrix
\begin{equation}
	\label{eq:k_prdm_trace_normalised}
	\gamma^{(k)} \coloneqq \Tr_{k+1\rightarrow N}\left|\Psi\right\rangle\left\langle\Psi\right| = {N \choose k}^{-1}\Gamma^{(k)}.
\end{equation}
More specifically, the estimate \eqref{eq:k_prdm_hs_estimate} directly yields a lower on the \textit{von Neumann entropy}
\begin{equation}
	\label{eq:von_neumann_entropy}
	S(\gamma^{(k)}) \coloneqq -\Tr\left(\gamma^{(k)}\log \gamma^{(k)}\right)
\end{equation}
and the $2$-\textit{Rényi entropy}
\begin{equation*}
	S_2(\gamma^{(k)}) \coloneqq - \Tr\left(\log \left[\left(\gamma^{(k)}\right)^2\right]\right).
\end{equation*}
Put on formal grounds, this corresponds to the following corollary.
\begin{corollary}
	\label{corollary:entropy_bounds}
	Let $\Psi \in \mathfrak{H}^{\wedge N}$ be normalised and define its trace normalised $k$-particle reduced density matrix $\gamma^{(k)}$ as in \eqref{eq:k_prdm_trace_normalised}. Then
	\begin{equation}
		\label{eq:K-particle_density_von_neumann_entropy_estimate}
		S\left(\gamma^{(k)}\right) \geq k\log N + \mathcal{O}(1)
	\end{equation}
	and
	\begin{equation}
		\label{eq:K-particle_density_2_reyni_entropy_estimate}
		S_2\left(\gamma^{(k)}\right) \geq k\log N + \mathcal{O}(1).
	\end{equation}
\end{corollary}

\begin{proof}
	As pointed out in \cite{Carlen2016EntropyEnanglement}, Jensen's inequality applied to the convex function $x\mapsto -\log(x)$ implies
	\begin{equation*}
		S\left(\gamma^{(k)}\right) \geq -\log\left(\Vert\gamma^{(k)}\Vert_\textmd{HS}^2\right),
	\end{equation*}
	which when combined with \eqref{eq:k_prdm_hs_estimate} yields \eqref{eq:K-particle_density_von_neumann_entropy_estimate}. The estimate \eqref{eq:K-particle_density_2_reyni_entropy_estimate} follows in an analogous way.
\end{proof}
		
	\section{Notation and preliminaries}
	\subsection{Antisymmetry and sign conventions}

For $\psi\in\mathfrak{H}^{\wedge N}$ and $\varphi\in\mathfrak{H}^{\wedge M}$, we define the antisymmetric tensor product $\psi\wedge\varphi\in\mathfrak{H}^{\wedge(N + M)}$ as
\begin{equation*}
	\psi\wedge\varphi(x_1,\dots,x_{N + M}) \coloneqq \dfrac{1}{\sqrt{N!M!(N + M)!}}\sum_{\sigma\in\mathcal{S}_{N + M}}\sgn(\sigma) U_\sigma(\psi\otimes\varphi),
\end{equation*}
where $\mathcal{S}_{N + M}$ denotes the group of permutations of $\{1,\dots,N + M\}$, and $U_\sigma$ denotes the permutation operator defined in \eqref{eq:permutation_operator}. Given an orthonormal family $(u_i)_{i\geq 1}$ in $\mathfrak{H}$ and a multi-index $\bfA = (a_1,\dots,a_N)$, we use the short-hand notation
\begin{equation*}
	u_{\bfA} \coloneqq u_{a_1}\wedge \dots\wedge u_{a_N},
\end{equation*}
where $a_1 < \dots < a_N$. Throughout this paper, we use bold letters such as $\bfA$ to denote multi-indices, which we always assume to be strictly ordered sequences of distinct integers. Given two multi-indices $\bfA$ and $\bfB$ of sizes $N$ and $M$, respectively, we define the union $\bfA\cup\bfB$ as the ordered multiset (allowing repetitions) of size $N + M$ that includes all elements of $\bfA$ and $\bfB$. For example, if $\bfA = (1,3)$ and $\bfB = (2,3)$, then $\bfA\cup\bfB = (1,2,3,3)$. Though we will almost exclusively use $\bfA\cup\bfB$ when $\bfA$ and $\bfB$ are disjoint, defining it more generally will simplify some notation. We also take the intersection $\bfA\cap\bfB$ and the set difference $\bfA\setminus\bfB$ to be ordered.

Now, let $\bfalpha = (\alpha_1,\dots,\alpha_s)$ and $\bfbeta = (\beta_1,\dots,\beta_t)$ be two disjoint multi-indices (with $\alpha_1 < \dots < \alpha_s$ and $\beta_1 < \dots < \beta_t$ by convention). Define $\sgn(\bfalpha,\bfbeta)$ as the sign of the permutation that orders the concatenated multi-index $(\alpha_1,\dots,\alpha_s,\beta_1,\dots,\beta_t)$. We refer to this as the \textit{relative sign} of $\bfalpha$ and $\bfbeta$. For example, if $\bfalpha = (1,3)$ and $\bfbeta = (2,4)$, then the concatenation is $(1,3,2,4)$, and it takes a single transposition $(2\; 3)$ to sort it into $(1,2,3,4)$, so $\sgn(\bfalpha,\bfbeta) = -1$. Recall that, if a permutation is expressed as a product of transpositions (or other permutations), its sign is equal to the product of the signs of its components.

Since $\bfalpha$ and $\bfbeta$ are ordered individually, we can compute $\sgn(\bfalpha,\bfbeta)$ recursively:
\begin{equation}
	\label{eq:sign_permutation_equals_product}
	\sgn(\bfalpha,\bfbeta) = \prod_{i = 1}^s\sgn(\alpha_i,\bfbeta).
\end{equation}
Let us explain this briefly. To insert $\alpha_s$ into the correct place among the elements of $\bfbeta$, we count how many elements of $\bfbeta$ are less than $\alpha_s$. That number determines the sign $\sgn(\alpha_s,\bfbeta)$, as $\alpha_s$ must be transposed past that many elements. Once $\alpha_s$ is placed, we repeat the process for $\alpha_{s -1},\alpha_{s-2},$ etc. Because the $\alpha_i$'s are ordered, we never swap them with one another, and each $\alpha_i$ is inserted independently. This recursive insertion justifies \eqref{eq:sign_permutation_equals_product}.

A direct consequence of \eqref{eq:sign_permutation_equals_product} is the identity
\begin{equation}
	\label{eq:sign_permutation_reverse_order}
	\sgn(\bfalpha,\bfbeta) = 
	\left\{
	\begin{array}{ll}
		\sgn(\bfbeta,\bfalpha) &\textmd{if $\vert\bfalpha\vert$ or $\vert\bfbeta\vert$ is even,}\\
		-\sgn(\bfbeta,\bfalpha) &\textmd{if both $\vert\bfalpha\vert$ and $\vert\bfbeta\vert$ are odd.}
	\end{array}
	\right.
\end{equation}
Indeed, it takes $t$ transpositions to move each $\alpha_i$ past the entire block $\bfbeta$ when reordering $(\bfalpha,\bfbeta)$ into $(\bfbeta,\bfalpha)$. Doing so for all $s$ elements of $\bfalpha$ requires $st$ transitions, and thus
\begin{equation*}
	\sgn(\bfalpha,\bfbeta) = (-1)^{st}\sgn(\bfbeta,\bfalpha).
\end{equation*}

Another important property of the relative sign is that, given a third multi-index $\bfdelta$ disjoint from $\bfalpha$ and $\bfbeta$, we have
\begin{equation}
	\label{eq:sign_permutation_union_equals_product}
	\sgn(\bfalpha,\bfbeta\cup\bfdelta) = \sgn(\bfalpha,\bfbeta)\sgn(\bfalpha,\bfdelta).
\end{equation}
To verify this, we apply \eqref{eq:sign_permutation_equals_product}:
\begin{equation*}
	\sgn(\bfalpha,\bfbeta\cup\bfdelta) = \prod_{i = 1}^s\sgn(\alpha_i,\bfbeta\cup\bfdelta).
\end{equation*}
Then, we use that each $\sgn(\alpha_i,\bfbeta\cup\bfdelta)$ can be factored as
\begin{equation}
	\label{eq:sign_permutation_union_equals_product_intermediate}
	\sgn(\alpha_i,\bfbeta\cup\bfdelta) = \sgn(\alpha_i,\bfbeta)\sgn(\alpha_i,\bfdelta),
\end{equation}
which we now justify. To place $\alpha_i$ into $\bfbeta\cup\bfdelta$, it must be moved past each of the elements of $\bfbeta\cup\bfdelta$ that are less than it. Since $\bfbeta\cup\bfdelta$ is ordered, this means that $\alpha_i$ must be moved past the elements of $\bfbeta$ that are less than it, and the elements of $\bfdelta$ than are less than it. Elements of $\bfbeta$ and $\bfdelta$ must however never be swapped. Thus, the total total sign $\sgn(\alpha_i,\bfbeta\cup\bfdelta)$ is the product of the signs arising from the placement of $\alpha_i$ among $\bfbeta$ and among $\bfdelta$. This proves \eqref{eq:sign_permutation_union_equals_product_intermediate}, and hence establishes \eqref{eq:sign_permutation_union_equals_product}.

\subsection{Sums over multi-indices}

Throughout the paper, we shall often take sums over multi-indices $\bfA$ or $\bfB$ of a fixed length. Given a function $f:\N^N\rightarrow\C$, we denote
\begin{equation*}
	\sum_{\vert\bfA\vert = N}f(\bfA) \ceqq \sum_{a_1 < \dots <a_N}f(a_1,\dots,a_N).
\end{equation*}
Likewise, for a function $g:\N^N\times\N^N\rightarrow\C$, we write
\begin{equation*}
	\sum_{\vert\bfA\vert,\vert\bfB\vert =N}g(\bfA,\bfB) \ceqq \sum_{\substack{a_1 < \dots < a_N\\ b_1 < \dots < b_N}}g(a_1,\dots,a_N,b_1,\dots,b_N).
\end{equation*}
We provide here three elementary results that will be used regularly throughout the proof of Theorem~\ref{th:HS_estimates}. The proofs are given in Appendix~\ref{appendix:proof_combinatorial_results}.

\begin{lemma}
	\label{lemma:rewritting_sums}
	Let $N \leq M$, and let
	\begin{equation*}
		f:\N^N\times\N^M \rightarrow \C
	\end{equation*}
	be a function such that for all $\bfA\in \N^N$ and $\bfB\in\N^M$,
	\begin{equation}
		\label{eq:rewritting_sums_index_distinct_condition}
		f(\bfA,\bfB) = 0
	\end{equation}
	whenever either $\bfA$ or $\bfB$ contains repeated indices. Then,
	\begin{equation}
		\label{eq:rewritting_sums}
		\sum_{\vert\bfA\vert = N}\sum_{\vert\bfB\vert = M}f(\bfA,\bfB) = \sum_{r = 0}^N\sum_{\vert\bfD\vert = N - r}\sum_{\vert\bfA\vert = r}\sum_{\substack{\vert\bfB\vert = M - N + r\\ \bfA\cap\bfB = \emptyset}}f(\bfD\cup\bfA,\bfD\cup\bfB).
	\end{equation}
\end{lemma}

\begin{lemma}
	\label{lemma:rewritting_sums_2}
	Let $N$ and $M$ be two integers, and let
	\begin{equation*}
		f: \N^N\times\N^M \rightarrow \C
	\end{equation*}
	be a function such that for all $\bfA\in\N^N$, $\bfB\in\N^M$,
	\begin{equation*}
		f(\bfA,\bfB) = 0
	\end{equation*}
	whenever $\bfA\cup\bfB$ contains repeated indices. Then,
	\begin{equation}
		\label{eq:rewritting_sums_3}
		\sum_{\vert\bfA\vert = N}\sum_{\vert\bfB\vert = M}f(\bfA,\bfB) = \sum_{\vert\bfA\vert = N + M}\sum_{\substack{\vert\bfB\vert = M\\ \bfB\subset \bfA}}f(\bfA\setminus\bfB,\bfB).
	\end{equation}
\end{lemma}

\begin{lemma}
	\label{lemma:rewritting_sums_3}
	Let $N$ and $M$ be two integers, and let
	\begin{equation*}
		f:\N^{N + M} \rightarrow \C
	\end{equation*}
	be a function such that for all $\bfA\in\N^{N + M}$,
	\begin{equation}
		f(\bfA) = 0
	\end{equation}
	whenever $\bfA$ contains repeated indices. Then,
	\begin{equation}
		\sum_{\vert\bfA\vert = N}\sum_{\vert\bfB\vert = M}f(\bfA\cup\bfB) = {N + M \choose N}\sum_{\vert\bfA\vert = N + M}f(\bfA).
	\end{equation}
\end{lemma}
	
	\section{Expansion into Slater determinants and main estimate}
	\subsection{Rewriting of the Hilbert--Schmidt norm}

The first step in the proof of Theorem~\ref{th:HS_estimates} is the following rewriting of the Hilbert--Schmidt norm of $\Gamma^{(k)}$. 

\begin{lemma}
	\label{lemma:rewriting_hs_norm}
	Let $\Psi\in\mfH^{\wedge N}$ be a normalised state. Let $(u_i)_{i \geq 1}$ be an orthonormal basis of $\mfH$, and expand $\Psi$ into Slater determinants built from $(u_i)_{i\geq 1}$:
	\begin{equation}
		\label{eq:wavefunction_expanded_basis_slater}
		\Psi = \sum_{\vert \bfA\vert = N}c_{\bfA}u_{\bfA},
	\end{equation}
	where
	\begin{equation}
		\label{eq:convention_coef_slater_expansion}
		c_{\bfA} = 0 \textmd{ if $\bfA$ contains the same index more than once.}
	\end{equation}
	Define
	\begin{equation}
		\label{eq:function_lamba_def}
		\Lambda(\bfD;\bfalpha,\bfbeta;\bfeps,\bfeta) = \sgn(\bfalpha\cup\bfbeta,\bfeps\cup\bfeta)c_{\bfD\cup\bfalpha\cup\bfeps}\overline{c_{\bfD\cup\bfalpha\cup\bfeta}c_{\bfD\cup\bfbeta\cup\bfeps}}c_{\bfD\cup\bfbeta\cup\bfeta},
	\end{equation}
	with the convention that $\Lambda(\bfD;\bfalpha,\bfbeta;\bfeps,\bfeta) = 0$ unless all involved multi-indices are pairwise disjoint. Set
	\begin{equation*}
		\Lambda(\bfD;\bfeps,\bfeta) = \Lambda(\bfD;\emptyset,\emptyset;\bfeps,\bfeta).
	\end{equation*}
	Then, the Hilbert--Schmidt norm of $\Gamma^{(k)}$ can be written as
	\begin{equation}
		\label{eq:k_prdm_hilbert_schmidt_norm_rewritten}
		\Vert\Gamma^{(k)}\Vert_{\HS}^2\\
		= \sum_{s = 0}^k\sum_{r = 0}^N{N - r \choose k - s}\sum_{\vert \bfD\vert = N - r}\sum_{\vert\bfeps\vert, \vert\bfeta\vert = r - s}\sum_{\vert\bfalpha\vert, \vert\bfbeta\vert = s}\Lambda(\bfD;\bfalpha,\bfbeta;\bfeps,\bfeta).
	\end{equation}
\end{lemma}

Let us briefly comment on the expression \eqref{eq:k_prdm_hilbert_schmidt_norm_rewritten} before providing its proof. By expanding $\Psi$ as in \eqref{eq:wavefunction_expanded_basis_slater}, we get
\begin{equation*}
	\Vert\Gamma^{(k)}\Vert_{\HS}^2 = \sum_{\substack{\bfA,\bfB\\ \bfA',\bfB'}}\overline{c_\bfA}c_\bfB c_{\bfA'}\overline{c_{\bfB'}}\Tr\big(\Tr_{k+1\rightarrow N}\vert u_\bfA\rangle\langle u_\bfB\vert\Tr_{k+1\rightarrow N}\vert u_{\bfB'}\rangle\langle u_{\bfA'}\vert\big).
\end{equation*}
Then, $\bfA,\bfB,\bfA'$ and $\bfB'$ can be related to the indices in \eqref{eq:k_prdm_hilbert_schmidt_norm_rewritten} as follows (illustrated in Figure~\ref{fig:illustration_hs_norm_decomposition}):
\begin{itemize}[label=-]
	\item $\bfD$ is the set of indices that $\bfA,\bfB,\bfA'$ and $\bfB'$ all have in common;
	\item $\bfeps$ is the set of indices that only $\bfA$ and $\bfB$ have in common;
	\item $\bfeta$ is the set of indices that only $\bfA'$ and $\bfB'$ have in common;
	\item $\bfalpha$ is the set of indices that only $\bfA$ and $\bfA'$ have in common;
	\item $\bfbeta$ is the set of indices that only $\bfB'$ and $\bfB'$ have in common.
\end{itemize}

\begin{figure}[ht]
	\centering
	\hfill
	\begin{tikzpicture}
	\draw[very thick, black] (0,0) -- (4,0);
	\foreach \i in {0,...,8}
	{
		\draw[thick,gray] (\i*0.5,-0.1)--(\i*0.5,0.1);
	}

	\draw[very thick, black] (0,-0.5) -- (1.5,-0.5);
	\draw[very thick, black] (2.5,-0.5) -- (5,-0.5);
	\foreach \i in {0,...,3}
	{
		\draw[thick,gray] (\i*0.5,-0.5 -0.1)--(\i*0.5,-0.5 + 0.1);
	}
	\foreach \i in {0,...,5}
	{
		\draw[thick,gray] (\i*0.5 + 2.5,-0.5 -0.1)--(\i*0.5 + 2.5,-0.5 + 0.1);
	}

	\draw[very thick, black] (0,-1) -- (2.5,-1);
	\foreach \i in {0,...,5}
	{
		\draw[thick,gray] (\i*0.5,-1 -0.1)--(\i*0.5,-1 + 0.1);
	}
	\draw[very thick, black] (5,-1) -- (6.5,-1);
	\foreach \i in {0,...,3}
	{
		\draw[thick,gray] (\i*0.5 + 5,-1 -0.1)--(\i*0.5 + 5,-1 + 0.1);
	}

	\draw[very thick, black] (0,-1.5) -- (1.5,-1.5);
	\foreach \i in {0,...,3}
	{
		\draw[thick,gray] (\i*0.5,-1.5 -0.1)--(\i*0.5,-1.5 + 0.1);
	}
	\draw[very thick, black] (4,-1.5) -- (6.5,-1.5);
	\foreach \i in {0,...,5}
	{
		\draw[thick,gray] (\i*0.5 + 4,-1.5 -0.1)--(\i*0.5 + 4,-1.5 + 0.1);
	}

	\node at (7,0) {\large $\mathbf{A}$};
	\node at (7,-0.5) {\large $\mathbf{B}$};
	\node at (7,-1) {\large $\mathbf{A}'$};
	\node at (7,-1.5) {\large $\mathbf{B}'$};
\end{tikzpicture}
	\hfill
	\begin{tikzpicture}
	\draw[very thick, black] (0,0) -- (1.5,0);
	\draw[very thick, Cerulean] (1.5,0) -- (2.5,0);
	\draw[very thick, Green] (2.5,0) -- (4,0);
	\foreach \i in {0,...,8}
	{
		\draw[thick,gray] (\i*0.5,-0.1)--(\i*0.5,0.1);
	}

	\draw[very thick, black] (0,-0.5) -- (1.5,-0.5);
	\draw[very thick, Green] (2.5,-0.5) -- (4,-0.5);
	\draw[very thick, RedOrange] (4,-0.5) -- (5,-0.5);
	\foreach \i in {0,...,3}
	{
		\draw[thick,gray] (\i*0.5,-0.5 -0.1)--(\i*0.5,-0.5 + 0.1);
	}
	\foreach \i in {0,...,5}
	{
		\draw[thick,gray] (\i*0.5 + 2.5,-0.5 -0.1)--(\i*0.5 + 2.5,-0.5 + 0.1);
	}

	\draw[very thick, black] (0,-1) -- (1.5,-1);
	\draw[very thick, Cerulean] (1.5,-1) -- (2.5,-1);
	\foreach \i in {0,...,5}
	{
		\draw[thick,gray] (\i*0.5,-1 -0.1)--(\i*0.5,-1 + 0.1);
	}
	\draw[very thick, Dandelion] (5,-1) -- (6.5,-1);
	\foreach \i in {0,...,3}
	{
		\draw[thick,gray] (\i*0.5 + 5,-1 -0.1)--(\i*0.5 + 5,-1 + 0.1);
	}

	\draw[very thick, black] (0,-1.5) -- (1.5,-1.5);
	\foreach \i in {0,...,3}
	{
		\draw[thick,gray] (\i*0.5,-1.5 -0.1)--(\i*0.5,-1.5 + 0.1);
	}
	\draw[very thick, RedOrange] (4,-1.5) -- (5,-1.5);
	\draw[very thick, Dandelion] (5,-1.5) -- (6.5,-1.5);
	\foreach \i in {0,...,5}
	{
		\draw[thick,gray] (\i*0.5 + 4,-1.5 -0.1)--(\i*0.5 + 4,-1.5 + 0.1);
	}

	\node at (7,0) {\large $\mathbf{A}$};
	\node at (7,-0.5) {\large $\mathbf{B}$};
	\node at (7,-1) {\large $\mathbf{A}'$};
	\node at (7,-1.5) {\large $\mathbf{B}'$};
\end{tikzpicture}
	\caption{The multi-indices in \eqref{eq:k_prdm_hilbert_schmidt_norm_rewritten} are colour coded in the right-hand side as follows: $\mathbf{D}$, ${\color{Cerulean}\bfalpha}$, ${\color{RedOrange}\bfbeta}$, ${\color{Green}\bfeps}$, ${\color{Dandelion}\bfeta}$.}
	\label{fig:illustration_hs_norm_decomposition}
\end{figure}

\begin{proof}[Proof of Lemma~\ref{lemma:rewriting_hs_norm}]
	As a consequence of \eqref{eq:wavefunction_expanded_basis_slater}, the $k$-particle reduced density matrix of $\Psi$ is given by
	\begin{equation}
		\label{eq:k_prdm_expanded_slater}
		\Gamma^{(k)} = {N \choose k}\sum_{\vert\bfA\vert, \vert\bfB\vert = N}c_{\bfA}\overline{c_{\bfB}}\Tr_{k+1\rightarrow N}\left\vert u_{\bfA}\right\rangle\left\langle u_{\bfB}\right\vert.
	\end{equation}
	Let us rewrite $\vert u_{\bfA}\rangle\langle u_{\bfB}\vert$. Since $\{u_{\bfalpha}: \vert \bfalpha\vert = k\}$ is an orthonormal basis of $\mathfrak{H}^{\wedge k}$, we have
	\begin{equation*}
		\sum_{\vert\bfalpha\vert = k}\left\vert u_{\bfalpha}\right\rangle\left\langle u_{\bfalpha}\right\vert = \mathds{1}
	\end{equation*}
	and, thus, we can write
	\begin{equation*}
		\Tr_{k + 1 \rightarrow N}\left\vert u_{\bfA}\right\rangle\left\langle u_{\bfB}\right\vert = \sum_{\substack{\vert\bfalpha\vert, \vert\bfbeta\vert = k}}\left\vert u_{\bfalpha}\right\rangle\left\langle u_{\bfalpha}\right\vert \left(\Tr_{k + 1 \rightarrow N}\left\vert u_{\bfA}\right\rangle\left\langle u_{\bfB}\right\vert\right) \left\vert u_{\bfbeta}\right\rangle\left\langle u_{\bfbeta}\right\vert.
	\end{equation*}
	By the orthonormality of $(u_i)_{i\geq 1}$, only the terms with $\bfalpha\subset \bfA$ and $\bfbeta \subset \bfB$ can be nonzero. In that case, by definition of the antisymmetric tensor product and of the relative sign, we have
	\begin{equation*}
		u_{\bfA} = \sgn(\bfalpha,\bfA\setminus\bfalpha)u_{\bfalpha}\wedge u_{\bfA\setminus\bfalpha} \quad \textmd{and} \quad u_{\bfB} = \sgn(\bfbeta,\bfB\setminus\bfbeta)u_{\bfbeta}\wedge u_{\bfB\setminus\bfbeta}.
	\end{equation*}
	Hence,
	\begin{multline*}
		\Tr_{k + 1 \rightarrow N}\left\vert u_{\bfA}\right\rangle\left\langle u_{\bfB}\right\vert = \dfrac{1}{k!(N - k)!N!}\sum_{\substack{\vert\bfalpha\vert, \vert\bfbeta\vert = k\\ \bfalpha \subset \bfA, \bfbeta \subset \bfB}}\sgn(\bfalpha,\bfA\setminus\bfalpha)\sgn(\bfbeta,\bfB\setminus\bfbeta)\sum_{\sigma,\pi\in\mathcal{S}_N}\sgn(\sigma)\sgn(\pi)\\
		\times\left\vert u_{\bfalpha}\right\rangle\left\langle u_{\bfalpha}\right\vert \left(\Tr_{k + 1 \rightarrow N}\left\vert U_{\sigma}(u_{\bfalpha}\otimes u_{\bfA\setminus\bfalpha})\right\rangle\left\langle U_{\pi}(u_{\bfbeta}\otimes u_{\bfB\setminus\bfbeta})\right\vert\right) \left\vert u_{\bfbeta}\right\rangle\left\langle u_{\bfbeta}\right\vert.
	\end{multline*}
	Again, due to the orthonormality of the $u_i$'s, the terms in the right-hand side can only be nonzero if
	\begin{equation*}
		\sigma(\{1,\dots,k\}) = \{1,\dots,k\}, \quad \pi(\{1,\dots,k\}) = \{1,\dots,k\} \quad \textmd{and} \quad \bfA\setminus\bfalpha = \bfB\setminus\bfbeta.
	\end{equation*}
	Furthermore, such $\sigma$'s and $\pi$'s satisfy
	\begin{equation*}
		U_{\sigma}(u_{\bfalpha}\otimes u_{\bfA\setminus\bfalpha}) = \sgn(\sigma)u_{\bfalpha}\otimes u_{\bfA\setminus\bfalpha} \quad \textmd{and} \quad U_{\pi}(u_{\bfbeta}\otimes u_{\bfB\setminus\bfbeta}) = \sgn(\pi)u_{\bfbeta}\otimes u_{\bfB\setminus\bfbeta}.
	\end{equation*}
	This finally implies
	\begin{equation}
		\label{eq:trace_slater_rewritten}
		\Tr_{k + 1 \rightarrow N}\left\vert u_{\bfA}\right\rangle\left\langle u_{\bfB}\right\vert = \dfrac{k!(N - k)!}{N!}\sum_{\substack{\vert\bfalpha\vert, \vert\bfbeta\vert = k\\ \bfalpha \subset \bfA, \bfbeta \subset \bfB\\ \bfA\setminus\bfalpha = \bfB\setminus\bfbeta}}\sgn(\bfalpha,\bfA\setminus\bfalpha)\sgn(\bfbeta,\bfB\setminus\bfbeta)\left\vert u_{\bfalpha}\right\rangle\left\langle u_{\bfbeta}\right\vert.
	\end{equation}
	
	Injecting \eqref{eq:trace_slater_rewritten} into \eqref{eq:k_prdm_expanded_slater} yields
	\begin{align*}
		\Gamma^{(k)} &= \sum_{\vert\bfA\vert, \vert\bfB\vert = N}c_{\bfA}\overline{c_{\bfB}}\sum_{\substack{\vert\bfalpha\vert, \vert\bfbeta\vert = k\\ \bfalpha \subset \bfA, \bfbeta \subset \bfB\\ \bfA\setminus\bfalpha = \bfB\setminus\bfbeta}}\sgn(\bfalpha,\bfA\setminus\bfalpha)\sgn(\bfbeta,\bfB\setminus\bfbeta)\left\vert u_{\bfalpha}\right\rangle\left\langle u_{\bfbeta}\right\vert\\
		&= \sum_{\vert\bfalpha\vert,\vert\bfbeta\vert=k}\bigg(\sum_{\substack{\vert\bfA\vert,\vert\bfB\vert = N\\ \bfalpha \subset \bfA, \bfbeta\subset \bfB\\ \bfA\setminus\bfalpha = \bfB\setminus\bfbeta}}\sgn(\bfalpha,\bfA\setminus\bfalpha)\sgn(\bfbeta,\bfB\setminus\bfbeta)c_{\bfA}\overline{c_{\bfB}}\bigg)\left\vert u_{\bfalpha}\right\rangle\left\langle u_{\bfbeta}\right\vert.
	\end{align*}
	Thanks to Lemma~\ref{lemma:rewritting_sums_2}, we have, at fixed $\bfalpha, \bfbeta$ and $\bfB$,
	\begin{equation*}
		\sum_{\substack{\vert\bfA\vert = N\\ \bfalpha\subset\bfA\\ \bfA\setminus\bfalpha = \bfB\setminus\bfbeta}}\sgn(\bfalpha,\bfA\setminus\bfalpha)c_\bfA = \sum_{\substack{\vert\bfA\vert = N - k\\ \bfA = \bfB\setminus\bfbeta}}\sgn(\bfalpha,\bfA)c_{\bfA\cup\bfalpha}.
	\end{equation*}
	Let us stress on the fact that we are summing over $\bfA$'s that are disjoint from $\bfalpha$ on the right-hand side. We do not write this explicitly because it is enforced by the convention \eqref{eq:convention_coef_slater_expansion}. Doing the same thing for $\bfB$ and $\bfbeta$, we obtain
	\begin{equation}
		\label{eq:k_prdm_rewritten}
		\Gamma^{(k)} = \sum_{\vert\boldsymbol{A}\vert = N - k}\sum_{\vert\bfalpha\vert, \vert\bfbeta\vert = k}\sgn(\bfalpha,\mathbf{A})\sgn(\bfbeta,\mathbf{A})c_{\mathbf{A}\cup\bfalpha}\overline{c_{\mathbf{A}\cup\bfbeta}}\left\vert u_{\bfalpha}\right\rangle\left\langle u_{\bfbeta}\right\vert.
	\end{equation}
	Here, $\bfA$ is disjoint from both $\bfalpha$ and $\bfbeta$, though the latter two are not necessarily disjoint from one another.
	
	From \eqref{eq:k_prdm_rewritten}, we deduce
	\begin{equation*}
		\Vert\Gamma^{(k)}\Vert_{\textmd{HS}}^2 = \sum_{\vert\mathbf{A}\vert, \vert\mathbf{B}\vert = N - k}\sum_{\vert\bfalpha\vert, \vert\bfbeta\vert = k}
		\sgn(\bfalpha,\mathbf{A})\sgn(\bfbeta,\mathbf{A})\sgn(\bfalpha,\mathbf{B})\sgn(\bfbeta,\mathbf{B})c_{\mathbf{A}\cup\bfalpha}\overline{c_{\mathbf{A}\cup\bfbeta}c_{\mathbf{B}\cup\bfalpha}}c_{\mathbf{B}\cup\bfbeta}.
	\end{equation*}
	Applying Lemma~\ref{lemma:rewritting_sums} to the sum over $(\bfalpha,\bfbeta)$, and then applying it again to the sum over $(\bfA,\bfB)$, we obtain
	\begin{multline*}
		\Vert\Gamma^{(k)}\Vert_{\textmd{HS}}^2 = \sum_{r = 0}^{N - k}\sum_{s = 0}^k\sum_{\vert\bfD\vert = N - k - r}\sum_{\vert\bfD'\vert = k - s}\sum_{\substack{\vert\bfeps\vert,\vert\bfeta\vert = r\\ \bfeps\cap\bfeta = \emptyset}}\sum_{\substack{\vert\bfalpha\vert,\vert\bfbeta\vert = s\\ \bfalpha\cap\bfbeta = \emptyset}}\sgn(\bfD\cup\bfalpha,\bfD'\cup\bfeps)\sgn(\bfD\cup\bfbeta,\bfD'\cup\bfeps)\\
		\times\sgn(\bfD\cup\bfalpha,\bfD'\cup\bfeta)\sgn(\bfD\cup\bfbeta,\bfD'\cup\bfeta)c_{\bfD\cup\bfD'\cup\bfalpha\cup\bfeps}\overline{c_{\bfD\cup\bfD'\cup\bfalpha\cup\bfeta}c_{\bfD\cup\bfD'\cup\bfbeta\cup\bfeps}}c_{\bfD\cup\bfD'\cup\bfbeta\cup\bfeta}.
	\end{multline*}
	This can however directly be written in a simpler form. In fact, thanks to \eqref{eq:sign_permutation_union_equals_product}, we have
	\begin{multline*}
		\sgn(\bfD\cup\bfalpha,\bfD'\cup\bfeps)\sgn(\bfD\cup\bfbeta,\bfD'\cup\bfeps)\sgn(\bfD\cup\bfalpha,\bfD'\cup\bfeta)\sgn(\bfD\cup\bfbeta,\bfD'\cup\bfeta)\\
		\begin{aligned}[t]
			&= \sgn(\bfalpha,\bfeps)\sgn(\bfbeta,\bfeps)\sgn(\bfalpha,\bfeta)\sgn(\bfbeta,\bfeta)\\
			&= \sgn(\bfalpha\cup\bfbeta,\bfeps\cup\bfeta),
		\end{aligned}
	\end{multline*}
	and thus
	\begin{equation*}
		\Vert\Gamma^{(k)}\Vert_{\textmd{HS}}^2 = \sum_{r = 0}^{N - k}\sum_{s = 0}^k\sum_{\vert\bfD\vert = N - k - r}\sum_{\vert\bfD'\vert = k - s}\sum_{\vert\bfeps\vert,\vert\bfeta\vert = r}\sum_{\vert\bfalpha\vert,\vert\bfbeta\vert = s}\Lambda(\bfD\cup\bfD';\bfalpha,\bfbeta;\bfeps,\bfeta).
	\end{equation*}
	As a result of the convention taken on $\Lambda$, the six multi-indices involved in the previous expression must be pairwise disjoint. Finally, an application of Lemma~\ref{lemma:rewritting_sums_3} implies
	\begin{equation*}
		\Vert\Gamma^{(k)}\Vert_{\textmd{HS}}^2
		= \sum_{s = 0}^k\sum_{r = 0}^{N - k}{N - r - s \choose k - s}\sum_{\vert \bfD\vert = N - r - s}\sum_{\vert\bfeps\vert, \vert\bfeta\vert = r}\sum_{\vert\bfalpha\vert, \vert\bfbeta\vert = s}\Lambda(\bfD;\bfalpha,\bfbeta;\bfeps,\bfeta),
	\end{equation*}
	which, by a change of variables, concludes the proof of Lemma~\ref{lemma:rewriting_hs_norm}.
\end{proof}

\subsection{Main estimate}

\label{section:cancellation_main_order}

Before we present the main estimate of the paper, we explain the idea behind it. In an ideal world, we would be able to prove a bound such as
\begin{equation}
	\label{eq:main_estimate_ideal_world}
	\sum_{\vert\bfeps\vert,\vert\bfeta\vert = r - s}\sum_{\vert\bfalpha\vert,\vert\bfbeta\vert = s}\Lambda(\bfD;\bfalpha,\bfbeta;\bfeps,\bfeta) \leq {r \choose s}\sum_{\vert\bfeps\vert,\vert\bfeta\vert = r}\Lambda(\bfD;\bfeps,\bfeta),
\end{equation}
which, by the Vandermonde identity
\begin{equation*}
	\sum_{s = 0}^k{N - r \choose k - s}{r \choose s} = {N \choose k},
\end{equation*}
would imply the optimal estimate
\begin{equation*}
	\Vert\Gamma^{(k)}\Vert_{\HS}^2 \leq {N \choose k} \sum_{r = 0}^N\sum_{\vert\bfD\vert = N - r}\sum_{\vert\bfeps\vert,\vert\bfeta\vert = r}\Lambda(\bfD;\bfeps,\bfeta) = {N \choose k}.
\end{equation*}
Although the very strong bound \eqref{eq:main_estimate_ideal_world} holds for $s \leq 1$ (it is trivial for $s = 0 $ and follows from Lemma~\ref{lemma:main_estimate_odd_case} for $s = 1$), it unclear that such a bound should be true for $s \geq 2$. Moreover, even if such a bound were to hold it would presumably be very difficult to prove. Instead, we roughly speaking show that each term of the form
\begin{equation*}
	\sum_{\vert\bfD\vert = N - r}\sum_{\vert\bfeps\vert,\vert\bfeta\vert = r - s}\sum_{\vert\bfalpha\vert,\vert\bfbeta\vert = s}\Lambda(\bfD;\bfalpha,\bfbeta;\bfeps,\bfeta)
\end{equation*}
contributes at most $\mathcal{O}(N^s)$. Since these terms are weighted by combinatorial factors of order $\mathcal{O}(N^{k - s})$, we obtain an overall contribution of order $\mathcal{O}(N^k)$, as desired.

\begin{proposition}
	\label{prop:cancellation_main_order}
	Let $t$ be a nonnegative integer. Then, there exists a family of real coefficients
	\begin{equation*}
		\left\{C_{s,t}(r,N): s\in\{0,\dots,t\}, 0\leq r \leq N\right\}
	\end{equation*}
	and a nonnegative constant $C_t$ depending only on $t$ and satisfying
	\begin{equation}
		\label{eq:cancellation_main_order_condition_coefficients}
		\left\vert C_{s,t}(r,N)\right\vert \leq C_tN^{t - s},
	\end{equation}
	for all $0\leq r \leq N$, and such that, for all normalised $\Psi\in\mfH^{\wedge N}$ expanded as in Lemma~\ref{lemma:rewriting_hs_norm}, the following estimate holds:
	\begin{multline}
		\label{eq:cancellation_main_order}
		\sum_{r = 0}^{N}\sum_{\vert\bfD\vert = N - r - t}\sum_{\vert\bfeps\vert, \vert\bfeta\vert = r}\sum_{\vert\bfalpha\vert, \vert\bfbeta\vert = t}\Lambda(\bfD;\bfalpha,\bfbeta;\bfeps,\bfeta)\\
		\leq \sum_{s = 0}^{t - 1}\sum_{r = 0}^{N}C_{s,t}(r,N)\sum_{\vert\bfD\vert = N - r - s}\sum_{\vert\bfeps\vert, \vert\bfeta\vert = r}\sum_{\vert\bfalpha\vert, \vert\bfbeta\vert = s}\Lambda(\bfD;\bfalpha,\bfbeta;\bfeps,\bfeta).
	\end{multline} 
\end{proposition}

To prove Proposition~\ref{prop:cancellation_main_order}, we need to treat the cases $t$ odd and $t$ even separately, which is mainly a consequence of \eqref{eq:sign_permutation_reverse_order}. Since the estimate obtained in the odd case is simpler and much stronger, we begin by proving that one.

\subsection{Proof of Proposition~\ref{prop:cancellation_main_order} in the odd case}

The main part of the proof of Proposition~\ref{prop:cancellation_main_order} in the odd case is contained in the following lemma.

\begin{lemma}
	\label{lemma:main_estimate_odd_case}
	Expand $\Psi$ into Slater determinants built from some orthonormal basis $(u_i)_{i_ \geq 1}$ of $\mfH$ as in Lemma~\ref{lemma:rewriting_hs_norm}. Define $\Lambda$ as in Lemma~\ref{lemma:rewriting_hs_norm}. Let $t$ be a nonnegative integer. Then, for any $t\leq r\leq N$ and for any $\bfD$ satisfying $\vert\bfD\vert = N - r$,
	\begin{multline}
		\label{eq:main_estimate_odd_case}
		\sum_{\substack{\vert\bfeps\vert,\vert\bfeta\vert = r - t}}\sum_{\substack{\vert\bfalpha\vert,\vert\bfbeta\vert = t}}\Lambda(\bfD;\bfalpha,\bfbeta;\bfeps,\bfeta) \leq \dfrac{1}{2}{r \choose t}\sum_{\substack{\vert\bfeps\vert,\vert\bfeta\vert = r}}\Lambda(\bfD;\bfeps,\bfeta)\\
		+ \dfrac{1}{2}\sum_{s = 0}^t(-1)^s{r - s \choose t}{r - s \choose t - s}\sum_{\substack{\vert\bfeps\vert,\vert\bfeta\vert = r - s}}\sum_{\substack{\vert\bfalpha\vert,\vert\bfbeta\vert = s}}\Lambda(\bfD;\bfalpha,\bfbeta;\bfeps,\bfeta).
	\end{multline}
\end{lemma}

Although Lemma~\ref{lemma:main_estimate_odd_case} is true regardless of the parity of $t$, it is particularly strong only in the odd case. To see why, notice that the right-hand side contains the same term as the left-hand side with a prefactor $(-1)^t{r - t \choose t}/2$. In the odd case, this allows us to bring this term to the left of the inequality and divide by a factor of order $r^t$, which yields a strong bound. In the even case, however, this argument does not work since the prefactor comes with a plus sign.

To shorten some expressions, we define
\begin{equation}
	\label{eq:main_estimate_odd_case_def_xi}
	\Xi(\bfD;\bfalpha,\bfbeta;\bfeps,\bfeta) \ceqq c_{\bfD\cup\bfalpha\cup\bfeps}\overline{c_{\bfD\cup\bfalpha\cup\bfeta}c_{\bfD\cup\bfbeta\cup\bfeps}}c_{\bfD\cup\bfbeta\cup\bfeta}.
\end{equation}
Unlike $\Lambda$ however, we do \textit{not} take the convention that $\Xi$ vanishes when the multi-indices are not disjoint.

\begin{proof}[Proof of Lemma~\ref{lemma:main_estimate_odd_case}]
	For readability's sake, we assume that the coefficients $c_{\bfA}$ are real. The same proof can be applied when dealing with complex coefficients by appropriately adding complex conjugates and moduli.
	
	We begin the proof by applying Lemma~\ref{lemma:rewritting_sums_2} twice: first to the pair $(\bfalpha,\bfeps)$, and then to $(\bfbeta,\bfeta)$. This gives
	\begin{equation}
		\label{eq:k_prdm_main_estimate_proof_odd_case_rewritten}
		\sum_{\vert\bfeps\vert,\vert\bfeta\vert = r - t}\sum_{\vert\bfalpha\vert,\vert\bfbeta\vert = t}\Lambda(\bfD;\bfalpha,\bfbeta;\bfeps,\bfeta) = \sum_{\vert\bfeps\vert,\vert\bfeta\vert = r}\sum_{\substack{\vert\bfalpha\vert,\vert\bfbeta\vert = t\\ \bfalpha\subset\bfeps, \bfbeta\subset \bfeta}}\Lambda(\bfD;\bfalpha,\bfbeta;\bfeps\setminus\bfalpha,\bfeta\setminus\bfbeta).
	\end{equation}
	Let us emphasise that we are summing over pairwise disjoint multi-indices in both the left and the right of the equality; this is implied by the assumption that $\Lambda$ vanishes whenever the multi-indices are not pairwise disjoint. Thanks to \eqref{eq:sign_permutation_reverse_order} and \eqref{eq:sign_permutation_union_equals_product}, the sign contained in the right-hand side can be rewritten as
	\begin{align*}
		\sgn(\bfalpha\cup\bfbeta,(\bfeps\setminus\bfalpha)\cup(\bfeta\setminus\bfbeta)) &= \sgn(\bfalpha,\bfbeta)\sgn(\bfbeta,\bfalpha)\sgn(\bfalpha,\bfeta\cup\bfeps\setminus\bfalpha)\sgn(\bfbeta,\bfeps\cup\bfeta\setminus\bfbeta)\\
		&= (-1)^t\sgn(\bfalpha,\bfeta\cup\bfeps\setminus\bfalpha)\sgn(\bfbeta,\bfeps\cup\bfeta\setminus\bfbeta).
	\end{align*}
	Hence, \eqref{eq:k_prdm_main_estimate_proof_odd_case_rewritten} becomes
	\begin{multline*}
		\sum_{\vert\bfeps\vert,\vert\bfeta\vert = r - t}\sum_{\vert\bfalpha\vert,\vert\bfbeta\vert = t}\Lambda(\bfD;\bfalpha,\bfbeta;\bfeps,\bfeta)\\
		= (-1)^t\sum_{\substack{\vert\bfeps\vert,\vert\bfeta\vert = r\\ \bfeps\cap\bfeta = \emptyset}}\sum_{\substack{\vert\bfalpha\vert = t\\ \bfalpha\subset\bfeps}}\sgn(\bfalpha,\bfeta\cup\bfeps\setminus\bfalpha)c_{\bfD\cup\bfeps}c_{\bfD\cup\bfeta}\sum_{\substack{\vert\bfbeta\vert = t\\ \bfbeta\subset \bfeta}}
		\begin{multlined}[t]
			\sgn(\bfbeta,\bfeps\cup\bfeta\setminus\bfbeta)\\
			\times c_{\bfD\cup\bfalpha\cup\bfeta\setminus\bfbeta}c_{\bfD\cup\bfbeta\cup\bfeps\setminus\bfalpha}.
		\end{multlined}
	\end{multline*}
	Then, using Young's inequality, we obtain
	\begin{multline}
		\label{eq:k_prdm_main_estimate_proof_odd_case_young_inequality}
		\sum_{\vert\bfeps\vert,\vert\bfeta\vert = r - t}\sum_{\vert\bfalpha\vert,\vert\bfbeta\vert = t}\Lambda(\bfD;\bfalpha,\bfbeta;\bfeps,\bfeta) \leq \dfrac{1}{2}\sum_{\substack{\vert\bfeps\vert,\vert\bfeta\vert = r\\ \bfeps\cap\bfeta = \emptyset}}\sum_{\substack{\vert\bfalpha\vert = t\\ \bfalpha\subset\bfeps}}c_{\bfD\cup\bfeps}^2c_{\bfD\cup\bfeta}^2\\
		\phantom{\leq} + \dfrac{1}{2}\sum_{\substack{\vert\bfeps\vert,\vert\bfeta\vert = r\\ \bfeps\cap\bfeta = \emptyset}}\sum_{\substack{\vert\bfalpha\vert = t\\ \bfalpha\subset\bfeps}}\bigg(\sum_{\substack{\vert\bfbeta\vert = t\\ \bfbeta\subset \bfeta}}\sgn(\bfbeta,\bfeps\cup\bfeta\setminus\bfbeta)c_{\bfD\cup\bfalpha\cup\bfeta\setminus\bfbeta}c_{\bfD\cup\bfbeta\cup\bfeps\setminus\bfalpha}\bigg)^2
	\end{multline}
		
	On the one hand, we have
	\begin{equation}
		\label{eq:k_prdm_main_estimate_proof_odd_case_first_term}
		\sum_{\substack{\vert\bfeps\vert,\vert\bfeta\vert = r\\ \bfeps\cap\bfeta = \emptyset}}\sum_{\substack{\vert\bfalpha\vert = t\\ \bfalpha\subset\bfeps}}c_{\bfD\cup\bfeps}^2c_{\bfD\cup\bfeta}^2 = {r \choose t}\sum_{\vert\bfeps\vert,\vert\bfeta\vert = r}\Lambda(\bfD;\bfeps,\bfeta).
	\end{equation}
	This matches the first term in \eqref{eq:main_estimate_odd_case}.
	
	On the other hand, we can develop
	\begin{multline}
		\label{eq:k_prdm_main_estimate_proof_odd_case_second_term}
		\bigg(\sum_{\substack{\vert\bfbeta\vert = t\\ \bfbeta\subset \bfeta}}\sgn(\bfbeta,\bfeps\cup\bfeta\setminus\bfbeta)c_{\bfD\cup\bfalpha\cup\bfeta\setminus\bfbeta}c_{\bfD\cup\bfbeta\cup\bfeps\setminus\bfalpha}\bigg)^2\\
		= \sum_{\substack{\vert\bfbeta\vert,\vert\bfbeta'\vert = t\\ \bfbeta,\bfbeta'\subset \bfeta}}\sgn(\bfbeta,\bfeps\cup\bfeta\setminus\bfbeta)\sgn(\bfbeta',\bfeps\cup\bfeta\setminus\bfbeta')\Xi(\bfD;\bfbeta,\bfbeta';\bfeps\setminus\bfalpha,\bfalpha\cup\bfeta\setminus(\bfbeta\cup\bfbeta')),
	\end{multline}
	where we recall that $\Xi$ was defined in \eqref{eq:main_estimate_odd_case_def_xi}. After an application of Lemma~\ref{lemma:rewritting_sums}, we get
	\begin{multline}
		\label{eq:k_prdm_main_estimate_proof_odd_case_second_term_not_rewritten}
		\bigg(\sum_{\substack{\vert\bfbeta\vert = t\\ \bfbeta\subset \bfeta}}\sgn(\bfbeta,\bfeps\cup\bfeta\setminus\bfbeta)c_{\bfD\cup\bfalpha\cup\bfeta\setminus\bfbeta}c_{\bfD\cup\bfbeta\cup\bfeps\setminus\bfalpha}\bigg)^2\\
		= \sum_{s = 0}^t\sum_{\substack{\vert\bfbeta\vert = t - s\\ \bfbeta\subset\bfeta}}\sum_{\substack{\vert\bfbeta'\vert,\vert\bfbeta''\vert = s\\ \bfbeta'\cap\bfbeta'' = \emptyset\\ \bfbeta',\bfbeta''\subset \bfeta}}
		\begin{multlined}[t]
			\sgn(\bfbeta\cup\bfbeta',\bfeps\cup\bfeta\setminus(\bfbeta\cup\bfbeta'))\sgn(\bfbeta\cup\bfbeta'',\bfeps\cup\bfeta\setminus(\bfbeta\cup\bfbeta''))\\
			\times\Xi(\bfD;\bfbeta\cup\bfbeta',\bfbeta\cup\bfbeta'';\bfeps\setminus\bfalpha,\bfalpha\cup\bfeta\setminus(\bfbeta\cup\bfbeta'\cup\bfbeta'')).
		\end{multlined}
	\end{multline}
	Using \eqref{eq:sign_permutation_union_equals_product}, we can rewrite the first sign in \eqref{eq:k_prdm_main_estimate_proof_odd_case_second_term_not_rewritten} as
	\begin{align*}
		\sgn(\bfbeta\cup\bfbeta',\bfeps\cup\bfeta\setminus(\bfbeta\cup\bfbeta')) &= \sgn(\bfbeta',(\bfeps\cup\bfbeta)\cup(\bfeta\setminus(\bfbeta\cup\bfbeta'\cup\bfbeta'')))\sgn(\bfbeta,\bfeps\cup\bfeta\setminus\bfbeta)\\
		&\phantom{=} \times\sgn(\bfbeta,\bfbeta')\sgn(\bfbeta',\bfbeta)\sgn(\bfbeta',\bfbeta''),
	\end{align*}
	and likewise for the second sign in \eqref{eq:k_prdm_main_estimate_proof_odd_case_second_term_not_rewritten}. As a consequence of \eqref{eq:sign_permutation_reverse_order}, $\sgn(\bfbeta,\bfbeta')\sgn(\bfbeta',\bfbeta)$ depends only on the sizes of $\bfbeta$ and $\bfbeta'$. Since $\bfbeta'$ and $\bfbeta''$ have the same size, this implies
	\begin{equation*}
		\sgn(\bfbeta,\bfbeta')\sgn(\bfbeta',\bfbeta)\sgn(\bfbeta,\bfbeta'')\sgn(\bfbeta'',\bfbeta) = 1.
	\end{equation*}
	Moreover, since $\bfbeta'$ and $\bfbeta''$ both have sizes $s$, the identity \eqref{eq:sign_permutation_reverse_order} also implies
	\begin{equation*}
		\sgn(\bfbeta',\bfbeta'')\sgn(\bfbeta'',\bfbeta') = (-1)^s.
	\end{equation*}
	Summing up,
	\begin{multline*}
		\sgn(\bfbeta\cup\bfbeta',\bfeps\cup\bfeta\setminus(\bfbeta\cup\bfbeta'))\sgn(\bfbeta\cup\bfbeta'',\bfeps\cup\bfeta\setminus(\bfbeta\cup\bfbeta''))\\
		= (-1)^s\sgn(\bfbeta',(\bfeps\cup\bfbeta)\cup(\bfeta\setminus(\bfbeta\cup\bfbeta'\cup\bfbeta'')))\sgn(\bfbeta'',(\bfeps\cup\bfbeta)\cup(\bfeta\setminus(\bfbeta\cup\bfbeta'\cup\bfbeta''))).
	\end{multline*}
	Injecting this into \eqref{eq:k_prdm_main_estimate_proof_odd_case_second_term_not_rewritten} yields
	\begin{multline*}
		\bigg(\sum_{\substack{\vert\bfbeta'\vert = t\\ \bfbeta'\subset \bfeta}}\sgn(\bfbeta',\bfeps\cup\bfeta\setminus\bfbeta')c_{\bfD\cup\bfalpha\cup\bfeta\setminus\bfbeta'}c_{\bfD\cup\bfbeta'\cup\bfeps\setminus\bfalpha}\bigg)^2\\
		= \sum_{s = 0}^t(-1)^s\sum_{\substack{\vert\bfbeta\vert = t - s\\ \bfbeta \subset\bfeta}}\sum_{\substack{\vert\bfbeta'\vert,\vert\bfbeta''\vert = s\\ \bfbeta',\bfbeta''\subset\bfeta}}\Lambda(\bfD;\bfbeta',\bfbeta'';\bfbeta\cup\bfeps\setminus\bfalpha,\bfalpha\cup\bfeta\setminus(\bfbeta\cup\bfbeta'\cup\bfbeta'')).
	\end{multline*}
	To conclude the proof of Lemma~\ref{lemma:main_estimate_odd_case}, we sum over $\bfeps,\bfeta$ and $\bfalpha$, and use four successive applications of Lemma~\ref{lemma:rewritting_sums_2} in order to get rid of the conditions $\bfalpha\subset\bfeps$ and $\bfbeta,\bfbeta',\bfbeta''\subset\bfeta$. A first application of Lemma~\ref{lemma:rewritting_sums_2} yields
	\begin{equation*}
		\sum_{\vert\bfeps\vert = r}\sum_{\substack{\vert\bfalpha\vert = t\\ \bfalpha\subset\bfeps}}\Lambda(\bfD;\bfbeta',\bfbeta'';\bfbeta\cup\bfeps\setminus\bfalpha,\bfalpha\cup\bfeta\setminus(\bfbeta\cup\bfbeta'\cup\bfbeta'')) = \sum_{\vert\bfeps\vert = r - t}\sum_{\vert\bfalpha\vert = t}\Lambda(\bfD;\bfbeta',\bfbeta'';\bfbeta\cup\bfeps;\bfalpha\cup\bfeta\setminus(\bfbeta\cup\bfbeta\cup\bfbeta'')).
	\end{equation*}
	Doing the same thing with $\bfeta$ and $\bfbeta$, then $\bfeta$ and $\bfbeta'$, and lastly with $\bfeta$ and $\bfbeta''$, we find
	\begin{multline*}
		\sum_{\substack{\vert\bfeps\vert,\vert\bfeta\vert = r\\ \bfeps\cap\bfeta = \emptyset}}\sum_{\substack{\vert\bfalpha\vert = t\\ \bfalpha\subset\bfeps}}\bigg(\sum_{\substack{\vert\bfbeta'\vert = t\\ \bfbeta'\subset \bfeta}}\sgn(\bfbeta',\bfeps\cup\bfeta\setminus\bfbeta')c_{\bfD\cup\bfalpha\cup\bfeta\setminus\bfbeta'}c_{\bfD\cup\bfbeta'\cup\bfeps\setminus\bfalpha}\bigg)^2\\
		= \sum_{s = 0}^t(-1)^s\sum_{\vert\bfeps\vert = r - t}\sum_{\vert\bfeta\vert = r - t - s}\sum_{\vert\bfalpha\vert = t}\sum_{\vert\bfbeta\vert = t - s}\sum_{\vert\bfbeta'\vert,\vert\bfbeta''\vert = s}\Lambda(\bfD;\bfbeta',\bfbeta'';\bfbeta\cup\bfeps,\bfalpha\cup\bfeta).
	\end{multline*}
	We now use Lemma~\ref{lemma:rewritting_sums_3} to regroup the sums over $\bfeps$ and $\bfbeta$ into a single sum, and likewise for $\bfeta$ and $\bfalpha$. The finally gives
	\begin{multline}
		\label{eq:k_prdm_main_estimate_proof_odd_case_second_term_estimate}
		\sum_{\substack{\vert\bfeps\vert,\vert\bfeta\vert = r\\ \bfeps\cap\bfeta = \emptyset}}\sum_{\substack{\vert\bfalpha\vert = t\\ \bfalpha\subset\bfeps}}\bigg(\sum_{\substack{\vert\bfbeta\vert = t\\ \bfbeta\subset \bfeta}}\sgn(\bfbeta,\bfeps\cup\bfeta\setminus\bfbeta)c_{\bfD\cup\bfalpha\cup\bfeta\setminus\bfbeta}c_{\bfD\cup\bfbeta\cup\bfeps\setminus\bfalpha}\bigg)^2\\
		= \sum_{s = 0}^t(-1)^s{r - s \choose t}{r - s \choose t - s}\sum_{\vert\bfeps\vert,\vert\bfeta\vert = r - s}\sum_{\vert\bfalpha\vert,\vert\bfbeta\vert = s}\Lambda(\bfD;\bfalpha,\bfbeta;\bfeps,\bfeta).
	\end{multline}
	Here, we renamed $\bfbeta'$ into $\bfalpha$ and $\bfbeta''$ into $\bfbeta$.
	
	Injecting the identities \eqref{eq:k_prdm_main_estimate_proof_odd_case_first_term} and \eqref{eq:k_prdm_main_estimate_proof_odd_case_second_term_estimate} into the estimate \eqref{eq:k_prdm_main_estimate_proof_odd_case_young_inequality} concludes the proof of Lemma~\ref{lemma:main_estimate_odd_case}.
	
\end{proof}
	
\begin{proof}[Proof of Proposition~\ref{prop:cancellation_main_order} in the odd case]
	We use Lemma~\ref{lemma:main_estimate_odd_case} to prove Proposition~\ref{prop:cancellation_main_order}. Notice that, in the right-hand side of \eqref{eq:main_estimate_odd_case}, the term corresponding to $s = t$ is the same as the one on the left-hand side with a prefactor $-{r - t \choose t}$ (the oddness of $t$ is crucial for the minus sign). Hence, we can simply shift this term the left of the equation, and divide both sides by $1 + {r - t \choose t}/2$ to obtain
	\begin{align*}
		\sum_{\vert\bfeps\vert,\vert\bfeta\vert = r - t}\sum_{\vert\bfalpha\vert,\vert\bfbeta\vert = t}\Lambda(\bfD;\bfalpha,\bfbeta;\bfeps,\bfeta) &\leq \dfrac{1}{2}\sum_{s = 0}^{t - 1}\dfrac{{r - s \choose t}{r - s \choose t - s}}{2 + {r - t \choose t}}(-1)^s\sum_{\vert\bfeps\vert,\vert\bfeta\vert  = r - s}\sum_{\vert\bfalpha\vert,\vert\bfbeta\vert = s}\Lambda(\bfD;\bfalpha,\bfbeta;\bfeps,\bfeta)\\
		&\phantom{\leq} + \dfrac{{r \choose t}}{2 + {r - t \choose t}}\sum_{\vert\bfeps\vert,\vert\bfeta\vert  = r}\Lambda(\bfD;\bfeps,\bfeta).
	\end{align*}
	The estimates
	\begin{equation*}
		\dfrac{{r \choose t}}{2 + {r - t \choose t}} \leq C_t \quad \textmd{and} \quad \dfrac{{r - s \choose t}{r - s \choose t - s}}{2 + {r - t \choose t}} \leq C_tN^{t - s}
	\end{equation*}
	show that we have proven \eqref{eq:cancellation_main_order} for $t$ odd.
\end{proof}

\subsection{Proof of Proposition~\ref{prop:cancellation_main_order} in the even case}

The core argument of the proof of Proposition~\ref{prop:cancellation_main_order} in the even case is contained in the following lemma.

\begin{lemma}
	\label{lemma:main_estimate_even_case}
	Expand $\Psi$ into Slater determinants built from some orthonormal basis $(u_i)_{i_ \geq 1}$ of $\mfH$ as in Lemma~\ref{lemma:rewriting_hs_norm}. Define $\Lambda$ as in \eqref{eq:function_lamba_def}. Let $t$ be an integer. Then, for any $t \leq r\leq N$ and for any $\bfD$ satisfying $\vert\bfD\vert = N - r$,
	\begin{multline}
		\label{eq:main_estimate_even_case}
		\sum_{\substack{\vert\bfeps\vert,\vert\bfeta\vert = r - t}}\sum_{\substack{\vert\bfalpha\vert,\vert\bfbeta\vert = t}}\Lambda(\bfD;\bfalpha,\bfbeta;\bfeps,\bfeta) \leq \sum_{s = 1}^{t - 1}(-1)^{s + t}{r - s \choose t - s}\sum_{\vert\bfeps\vert,\vert\bfeta\vert = r - s}\sum_{\vert\bfalpha\vert,\vert\bfbeta\vert = s}\Lambda(\bfD;\bfalpha,\bfbeta;\bfeps,\bfeta)\\
		\begin{aligned}[t]
			&\phantom{\leq} + \dfrac{1}{2t\tau}\sum_{s = 0}^{t + 1}\sum_{u = 0}^{\min(1,s)}(-1)^{s}{r - s \choose t - s + u}{r - s \choose 1 - u}\sum_{\vert\bfeps\vert,\vert\bfeta\vert=r - s}\sum_{\vert\bfalpha\vert,\vert\bfbeta\vert = s}\Lambda(\bfD;\bfalpha,\bfbeta;\bfeps,\bfeta)\\
			&\phantom{\leq} + \dfrac{\tau}{2t}\sum_{s = 0}^{t - 1}(-1)^s{r - s \choose t - s - 1}\sum_{\vert\bfeps\vert,\vert\bfeta\vert=r - s}\sum_{\vert\bfalpha\vert,\vert\bfbeta\vert = s}\Lambda(\bfD;\bfalpha,\bfbeta;\bfeps,\bfeta),
		\end{aligned}
	\end{multline}
	for all $\tau > 0$.
\end{lemma}

Just as in the proof of Lemma~\ref{lemma:main_estimate_odd_case}, we introduce
\begin{equation}
	\label{eq:main_estimate_even_case_def_xi}
	\Xi(\bfD;\bfalpha,\bfbeta;\bfeps,\bfeta) \ceqq c_{\bfD\cup\bfalpha\cup\bfeps}\overline{c_{\bfD\cup\bfalpha\cup\bfeta}c_{\bfD\cup\bfbeta\cup\bfeps}}c_{\bfD\cup\bfbeta\cup\bfeta}
\end{equation}
to shorten some notations. Again, we are \textit{not} imposing that $\Xi$ vanishes when the multi-indices are not disjoint.

\begin{proof}
	Assume for simplicity that the coefficients $c_\bfA$ are real. The generalisation to complex coefficients follows by adding complex conjugates and moduli where necessary.
	
	Using Lemma~\ref{lemma:rewritting_sums_2}, we can write
	\begin{equation*}
		\sum_{\vert\bfbeta\vert = t}\Lambda(\bfD;\bfalpha,\bfbeta;\bfeps,\bfeta) = t^{-1}\sum_{\vert\bfbeta\vert = t - 1}\sum_{\vert\bfdelta\vert = 1}\Lambda(\bfD;\bfalpha,\bfbeta\cup\bfdelta;\bfeps,\bfeta).
	\end{equation*}
	Note that the disjointness of $\bfbeta$ and $\bfdelta$ is insured by the assumption on $\Lambda$. Now, we apply Lemma~\ref{lemma:rewritting_sums} three times: once with $\bfalpha$ and $\bfeps$, then with $\bfbeta$ and the new $\bfeps$ (that now has length $r$), and lastly with $\bfdelta$ and $\bfeta$. This yields
	\begin{multline}
		\label{eq:k_prdm_main_estimate_proof_even_case_rewritten}
		\sum_{\substack{\vert\bfeps\vert,\vert\bfeta\vert = r - t}}\sum_{\substack{\vert\bfalpha\vert,\vert\bfbeta\vert = t}}\Lambda(\bfD;\bfalpha,\bfbeta;\bfeps,\bfeta)\\
		= \dfrac{1}{t}\sum_{\substack{\vert\bfeps\vert = r + t - 1\\ \vert\bfeta\vert = r - t + 1}}\sum_{\substack{\vert\bfalpha\vert = t\\ \bfalpha\subset\bfeps}}\sum_{\substack{\vert\bfbeta\vert = t - 1\\ \bfbeta\subset\bfeps}}\sum_{\substack{\vert\bfdelta\vert = 1\\ \bfdelta\subset\bfeta}}\Lambda(\bfD;\bfalpha,\bfbeta\cup\bfdelta;\bfeps\setminus(\bfalpha\cup\bfbeta),\bfeta\setminus\bfdelta).
	\end{multline}
	Using \eqref{eq:sign_permutation_reverse_order} and \eqref{eq:sign_permutation_union_equals_product}, as well as the fact that $t - 1$ and $t + 1$ have the same parity, we obtain
	\begin{multline*}
		\sgn(\bfalpha\cup\bfbeta\cup\bfdelta,\bfeps\cup\bfeta\setminus(\bfalpha\cup\bfbeta\cup\bfdelta))\\
		\begin{aligned}[t]
			&= \sgn(\bfalpha\cup\bfdelta,\bfbeta)\sgn(\bfbeta,\bfalpha\cup\bfbeta)\sgn(\bfalpha\cup\bfdelta,\bfeps\cup\bfeta\setminus(\bfalpha\cup\bfdelta))\sgn(\bfbeta,\bfeps\cup\bfeta\setminus\bfbeta)\\
			&= (-1)^{t + 1}\sgn(\bfalpha\cup\bfdelta,\bfeps\cup\bfeta\setminus(\bfalpha\cup\bfdelta))\sgn(\bfbeta,\bfeps\cup\bfeta\setminus\bfbeta).
		\end{aligned}
	\end{multline*}
	Hence,
	\begin{multline*}
		\sum_{\substack{\vert\bfalpha\vert = t\\ \bfalpha\subset\bfeps}}\sum_{\substack{\vert\bfbeta\vert = t - 1\\ \bfbeta\subset\bfeps}}\sum_{\substack{\vert\bfdelta\vert = 1\\ \bfdelta\subset\bfeta}}\Lambda(\bfD;\bfalpha,\bfbeta\cup\bfdelta;\bfeps\setminus(\bfalpha\cup\bfbeta),\bfeta\setminus\bfdelta)\\
		= (-1)^{t + 1}\sum_{\substack{\vert\bfbeta\vert = t - 1\\ \bfbeta\subset\bfeps}}\sgn(\bfbeta,\bfeps\cup\bfeta\setminus\bfbeta)c_{\bfD\cup\bfeps\setminus\bfbeta}c_{\bfD\cup\bfbeta\cup\bfeta}\sum_{\substack{\vert\bfalpha\vert = t\\ \bfalpha\subset\bfeps\setminus\bfbeta}}\sum_{\substack{\vert\bfdelta\vert = 1\\ \bfdelta\subset\bfeta}}
		\begin{multlined}[t]
			\sgn(\bfalpha\cup\bfdelta,\bfeps\cup\bfeta\setminus(\bfalpha\cup\bfdelta))\\
			\times c_{\bfD\cup\bfdelta\cup\bfeps\setminus\bfalpha}c_{\bfD\cup\bfalpha\cup\bfeta\setminus\bfdelta}.
		\end{multlined}
	\end{multline*}
	To bound this term, we would like to use Young's equality to separate the sum over $\bfbeta$ from the sums over $\bfalpha$ and $\bfdelta$. However, we currently cannot do so because of the disjointness condition $\bfalpha\cap\bfbeta = \emptyset$. To get around this issue, we add and remove the missing terms. Namely, we define
	\begin{equation}
		\label{eq:k_prdm_main_estimate_proof_even_case_first_term}
		\cI_1(\bfD;\bfeps,\bfeta) \ceqq \sum_{\substack{\vert\bfbeta\vert = t - 1\\ \bfbeta\subset\bfeps}}
		\begin{multlined}[t]
			\sgn(\bfbeta,\bfeps\cup\bfeta\setminus\bfbeta)c_{\bfD\cup\bfeps\setminus\bfbeta}c_{\bfD\cup\bfbeta\cup\bfeta}\\
			\times\sum_{\substack{\vert\bfalpha\vert = t\\ \bfalpha\subset\bfeps}}\sum_{\substack{\vert\bfdelta\vert = 1\\ \bfdelta\subset\bfeta}}\sgn(\bfalpha\cup\bfdelta,\bfeps\cup\bfeta\setminus(\bfalpha\cup\bfdelta))c_{\bfD\cup\bfdelta\cup\bfeps\setminus\bfalpha}c_{\bfD\cup\bfalpha\cup\bfeta\setminus\bfdelta}
		\end{multlined}
	\end{equation}
	and
	\begin{equation}
		\label{eq:k_prdm_main_estimate_proof_even_case_second_term}
		\cI_2(\bfD;\bfeps,\bfeta) \ceqq \sum_{\substack{\vert\bfbeta\vert = t - 1\\ \bfbeta\subset\bfeps}}\sum_{\substack{\vert\bfalpha\vert = t\\ \bfalpha\subset\bfeps\\ \bfalpha\cap\bfbeta \neq \emptyset}}\sum_{\substack{\vert\bfdelta\vert = 1\\ \bfdelta\subset\bfeta}}
		\begin{multlined}[t]
			\sgn(\bfbeta,\bfeps\cup\bfeta\setminus\bfbeta)\sgn(\bfalpha\cup\bfdelta,\bfeps\cup\bfeta\setminus(\bfalpha\cup\bfdelta))\\
			\times \Xi(\bfD;\bfalpha,\bfbeta\cup\bfdelta;\bfeps\setminus(\bfalpha\cup\bfbeta), \bfeta\setminus\bfdelta),
		\end{multlined}
	\end{equation}
	and write
	\begin{equation*}
		\sum_{\substack{\vert\bfalpha\vert = t\\ \bfalpha\subset\bfeps}}\sum_{\substack{\vert\bfbeta\vert = t - 1\\ \bfbeta\subset\bfeps}}\sum_{\substack{\vert\bfdelta\vert = 1\\ \bfdelta\subset\bfeta}}\Lambda(\bfD;\bfalpha,\bfbeta\cup\bfdelta;\bfeps\setminus(\bfalpha\cup\bfbeta),\bfeta\setminus\bfdelta) = (-1)^{t + 1}\cI_1(\bfD;\bfeps,\bfeta) + (-1)^t\cI_2(\bfD;\bfeps,\bfeta).
	\end{equation*}
	Recall that $\Xi$ was defined in \eqref{eq:main_estimate_even_case_def_xi}.
	
	As we now show, the term $\cI_2$ can be put into a much simpler form. Using Lemma~\ref{lemma:rewritting_sums}, we can write
	\begin{multline*}
		\cI_2(\bfD;\bfeps,\bfeta)\\
		= \sum_{s = 0}^{t - 2}\sum_{\substack{\vert\bfalpha\vert = t - 1 - s\\ \bfalpha \subset\bfeps}}\sum_{\substack{\vert\bfbeta'\vert = s\\ \vert\bfalpha'\vert = s + 1\\ \bfalpha'\cap\bfbeta' = \emptyset\\ \bfalpha',\bfbeta'\subset\bfeps}}
		\begin{multlined}[t]
			\sgn(\bfalpha\cup\bfbeta',\bfeps\cup\bfeta\setminus(\bfalpha\cup\bfbeta'))\sgn(\bfalpha\cup\bfalpha'\cup\bfdelta,\bfeps\cup\bfeta\setminus(\bfalpha\cup\bfalpha'\cup\bfdelta))\\
			\times \Xi(\bfD;\bfdelta\cup\bfbeta',\bfalpha';\bfeps\setminus(\bfalpha\cup\bfalpha'\cup\bfbeta'),\bfalpha\cup\bfeta\setminus\bfdelta).
		\end{multlined}
	\end{multline*}
	Using \eqref{eq:sign_permutation_reverse_order} and \eqref{eq:sign_permutation_union_equals_product}, we can show that
	\begin{multline*}
		\sgn(\bfalpha\cup\bfbeta',\bfeps\cup\bfeta\setminus(\bfalpha\cup\bfbeta'))\sgn(\bfalpha\cup\bfalpha'\cup\bfdelta,\bfeps\cup\bfeta\setminus(\bfalpha\cup\bfalpha'\cup\bfdelta))\\
		\begin{multlined}[t]
			= \sgn(\bfdelta\cup\bfalpha'\cup\bfbeta',(\bfalpha\cup\bfeta\setminus\bfbeta)\cup(\bfeps\setminus(\bfalpha\cup\bfalpha'\cup\bfbeta')))\\
			\times \sgn(\bfdelta\cup\bfalpha'\cup\bfbeta',\bfalpha)\sgn(\bfalpha,\bfdelta\cup\bfalpha'\cup\bfbeta')\sgn(\bfdelta\cup\bfalpha',\bfbeta')\sgn(\bfbeta',\bfdelta\cup\bfalpha').
		\end{multlined}
	\end{multline*}
	Since $\bfdelta\cup\bfalpha'\cup\bfbeta'$, $\bfdelta\cup\bfalpha'$ and $\bfbeta'$, respectively, have sizes $2(s + 1)$, $s + 2$ and $s$, we deduce
	\begin{multline*}
		\sgn(\bfalpha\cup\bfbeta',\bfeps\cup\bfeta\setminus(\bfalpha\cup\bfbeta'))\sgn(\bfalpha\cup\bfalpha'\cup\bfdelta,\bfeps\cup\bfeta\setminus(\bfalpha\cup\bfalpha'\cup\bfdelta))\\
		= (-1)^s\sgn(\bfdelta\cup\bfalpha'\cup\bfbeta',(\bfalpha\cup\bfeta\setminus\bfbeta)\cup(\bfeps\setminus(\bfalpha\cup\bfalpha'\cup\bfbeta'))).
	\end{multline*}
	Summing up, we have shown that
	\begin{equation*}
		\cI_2(\bfD;\bfeps,\bfeta) = \sum_{s = 0}^{t - 2}(-1)^s\sum_{\substack{\vert\bfalpha\vert = t - 1 - s\\ \bfalpha \subset \bfeps}}\sum_{\substack{\vert\bfbeta'\vert = s\\ \bfbeta' \subset \bfeps}}\sum_{\substack{\vert\bfalpha'\vert = s + 1\\ \bfalpha' \subset \bfeps}}\sum_{\substack{\vert\bfdelta\vert = 1\\ \bfdelta\subset\bfeta}}\Lambda(\bfD;\bfdelta\cup\bfbeta',\bfalpha';\bfeps\setminus(\bfalpha\cup\bfalpha'\cup\bfbeta'),\bfalpha\cup\bfeta\setminus\bfdelta).
	\end{equation*}
	We may now sum over $\bfeps$ and $\bfeta$, and apply Lemma~\ref{lemma:rewritting_sums_2} multiple times to obtain
	\begin{equation*}
		\sum_{\substack{\vert\bfeps\vert = r + t - 1\\ \vert\bfeta\vert = r - t + 1}}\cI_2(\bfD;\bfeps,\bfeta) = \sum_{s = 0}^{t - 2}(-1)^s\sum_{\substack{\vert\bfeps\vert = r - s -1\\ \vert\bfeta\vert = r - t}}\sum_{\vert\bfalpha\vert = t - s - 1}\sum_{\vert\bfbeta'\vert = s}\sum_{\vert \bfalpha'\vert = s + 1}\sum_{\vert\bfdelta\vert = 1}\Lambda(\bfD;\bfdelta\cup\bfbeta',\bfalpha';\bfeps,\bfalpha\cup\bfeta).
	\end{equation*}
	Using Lemma~\ref{lemma:rewritting_sums_3} and renaming the multi-indices, we find
	\begin{equation*}
		\sum_{\substack{\vert\bfeps\vert = r + t - 1\\ \vert\bfeta\vert = r - t + 1}}\cI_2(\bfD;\bfeps,\bfeta) = \sum_{s = 1}^{t - 1}(-1)^{s + 1}s{r - s \choose t - s}\sum_{\vert\bfeps\vert,\vert\bfeta\vert = r - s}\sum_{\vert\bfalpha\vert,\vert\bfbeta\vert = s}\Lambda(\bfD;\bfalpha,\bfbeta;\bfeps,\bfeta).
	\end{equation*}
	
	To sum up, we have so far proven that
	\begin{multline}
		\label{eq:k_prdm_main_estimate_proof_even_case_intermediate_summary}
		 \sum_{\substack{\vert\bfeps\vert,\vert\bfeta\vert = r - t}}\sum_{\substack{\vert\bfalpha\vert,\vert\bfbeta\vert = t}}\Lambda(\bfD;\bfalpha,\bfbeta;\bfeps,\bfeta) = \sum_{s = 1}^{t - 1}(-1)^{s + t}\dfrac{s}{t}{r - s \choose t - s}\sum_{\vert\bfeps\vert,\vert\bfeta\vert = r - s}\sum_{\vert\bfalpha\vert,\vert\bfbeta\vert = s}\Lambda(\bfD;\bfalpha,\bfbeta;\bfeps,\bfeta)\\
		 + \dfrac{(-1)^{t + 1}}{t}\sum_{\substack{\vert\bfeps\vert = r + t -1\\ \vert\bfeta\vert = r - t + 1}}\cI_1(\bfD;\bfeps,\bfeta),
	\end{multline}
	where $\cI_1$ is given by \eqref{eq:k_prdm_main_estimate_proof_even_case_first_term}. Next, we bound $\cI_1$. Defining
	\begin{equation*}
		\cJ_1(\bfD;\bfeps,\bfeta) \ceqq \bigg(\sum_{\substack{\vert\bfbeta\vert = t - 1\\ \bfbeta\subset\bfeps}}\sgn(\bfbeta,\bfeps\cup\bfeta\setminus\bfbeta)c_{\bfD\cup\bfeps\setminus\bfbeta}c_{\bfD\cup\bfbeta\cup\bfeta}\bigg)^2
	\end{equation*}
	and
	\begin{equation*}
		\cJ_2(\bfD;\bfeps,\bfeta) \ceqq \bigg(\sum_{\substack{\vert\bfalpha\vert = t\\ \bfalpha\subset\bfeps}}\sum_{\substack{\vert\bfdelta\vert = 1\\ \bfdelta\subset\bfeta}}\sgn(\bfalpha\cup\bfdelta,\bfeps\cup\bfeta\setminus(\bfalpha\cup\bfdelta))c_{\bfD\cup\bfdelta\cup\bfeps\setminus\bfalpha}c_{\bfD\cup\bfalpha\cup\bfeta\setminus\bfdelta}\bigg)^2,
	\end{equation*}
	and using Young's inequality, we have
	\begin{equation}
		\label{eq:k_prdm_main_estimate_proof_even_case_second_term_upper_bound}
		\cI_1(\bfD;\bfeps,\bfeta) \leq \dfrac{\tau}{2}\cJ_1(\bfD;\bfeps,\bfeta) + \dfrac{1}{2\tau}\cJ_2(\bfD;\bfeps,\bfeta),
	\end{equation}
	for all $\tau > 0$. \footnote{Roughly speaking, we are bounding a $t$-body term by the sum of a $(t-1)$-body term and a $(t + 1)$-body one. A similar idea was used by Christiansen in \cite{Christiansen2024HSEstimates}: bounding a two-body operator by a one-body operator and a three-body operator.}
	
	The term $\cJ_1$ can be expanded as
	\begin{equation*}
		\cJ_1(\bfD;\bfeps,\bfeta) = \sum_{\substack{\vert\bfbeta\vert,\vert\bfbeta'\vert = t - 1\\ \bfbeta,\bfbeta'\subset\bfeps}}\sgn(\bfbeta,\bfeps\cup\bfeta\setminus\bfbeta)\sgn(\bfbeta',\bfeps\cup\bfeta\setminus\bfbeta')\Xi(\bfD;\bfbeta,\bfbeta';\bfeps\setminus(\bfbeta\cup\bfbeta'),\bfeta).
	\end{equation*}
	This is more or less the same term as \eqref{eq:k_prdm_main_estimate_proof_odd_case_second_term}. Treating it in a very similar way yields
	\begin{equation}
		\label{eq:k_prdm_main_estimate_proof_even_case_second_term_first_part}
		\sum_{\substack{\vert\bfeps\vert = r + t - 1\\ \vert\bfeta\vert = r - t + 1}}\cJ_1(\bfD;\bfeps,\bfeta) = \sum_{s = 0}^{t - 1}(-1)^s{r - s \choose t - s}\sum_{\vert\bfeps\vert, \vert\bfeta\vert = r - s}\sum_{\vert\bfalpha\vert, \vert\bfbeta\vert = s}\Lambda(\bfD;\bfalpha,\bfbeta;\bfeps,\bfeta).
	\end{equation}
	
	The term $\cJ_2$ can be rewritten as
	\begin{equation*}
		\cJ_2(\bfD;\bfeps,\bfeta) = \sum_{\substack{\vert\bfalpha\vert,\vert\bfalpha'\vert = t\\ \bfalpha,\bfalpha'\subset\bfeps}}\sum_{\substack{\vert\bfdelta\vert,\vert\bfdelta'\vert = 1\\ \bfdelta,\bfdelta'\subset\bfeta}}
		\begin{multlined}[t]
			\sgn(\bfalpha\cup\bfdelta,\bfeps\cup\bfeta\setminus(\bfalpha\cup\bfdelta))\sgn(\bfalpha'\cup\bfdelta',\bfeps\cup\bfeta\setminus(\bfalpha'\cup\bfdelta'))\\
			\times\Xi(\bfD;\bfalpha\cup\bfdelta',\bfalpha'\cup\bfdelta;\bfeps\setminus(\bfalpha\cup\bfalpha'),\bfeta\setminus(\bfdelta\cup\bfdelta')).
		\end{multlined}
	\end{equation*}
	Using Lemma~\ref{lemma:rewritting_sums} with the pairs $(\bfalpha,\bfalpha')$ and $(\bfdelta,\bfdelta')$, we can expand this further:
	\begin{multline*}
		\cJ_2(\bfD;\bfeps,\bfeta)\\
		= \sum_{s = 0}^t\sum_{u = 0}^1\sum_{\substack{\vert\bfalpha\vert = t - s\\ \bfalpha\subset\bfeps}}\sum_{\substack{\vert\bfalpha'\vert,\vert\bfalpha''\vert=s\\\bfalpha',\bfalpha''\subset\bfeps\\ \bfalpha'\cap\bfalpha'' = \emptyset}}\sum_{\substack{\vert\bfdelta\vert = 1 - u\\ \bfdelta\subset\bfeta}}\sum_{\substack{\vert\bfdelta'\vert,\vert\bfdelta''\vert=u\\\bfdelta',\bfdelta''\subset\bfeta\\ \bfdelta'\cap\bfdelta'' = \emptyset}}
		\begin{multlined}[t]
			\sgn(\bfalpha\cup\bfalpha'\cup\bfdelta\cup\bfdelta',\bfeps\cup\bfeta\setminus(\bfalpha\cup\bfalpha'\cup\bfdelta\cup\bfdelta'))\\
			\times\sgn(\bfalpha\cup\bfalpha''\cup\bfdelta\cup\bfdelta'',\bfeps\cup\bfeta\setminus(\bfalpha\cup\bfalpha''\cup\bfdelta\cup\bfdelta''))\\
		\end{multlined}\\
		\times\Xi(\bfD;\bfalpha\cup\bfalpha'\cup\bfdelta\cup\bfdelta'',\bfalpha\cup\bfalpha''\cup\bfdelta\cup\bfdelta';\bfeps\setminus(\bfalpha\cup\bfalpha'\cup\bfalpha''),\bfeta\setminus(\bfdelta\cup\bfdelta'\cup\bfdelta'')).
	\end{multline*}
	With the help of \eqref{eq:sign_permutation_reverse_order} and \eqref{eq:sign_permutation_union_equals_product}, we find
	\begin{multline*}
		\begin{multlined}[t]
			\sgn(\bfalpha\cup\bfalpha'\cup\bfdelta\cup\bfdelta',\bfeps\cup\bfeta\setminus(\bfalpha\cup\bfalpha'\cup\bfdelta\cup\bfdelta'))\sgn(\bfalpha\cup\bfalpha''\cup\bfdelta\cup\bfdelta'',\bfeps\cup\bfeta\setminus(\bfalpha\cup\bfalpha''\cup\bfdelta\cup\bfdelta''))
		\end{multlined}\\
		=
		\begin{aligned}[t]
			(-1)^{s + u}&\sgn(\bfalpha'\cup\bfdelta'',(\bfeps\setminus(\bfalpha\cup\bfalpha'\cup\bfalpha'')\cup\bfeta\setminus(\bfdelta\cup\bfdelta'\cup\bfdelta''))\\
			&\times\sgn(\bfalpha''\cup\bfdelta',(\bfeps\setminus(\bfalpha\cup\bfalpha'\cup\bfalpha'')\cup\bfeta\setminus(\bfdelta\cup\bfdelta'\cup\bfdelta'')).
		\end{aligned}
	\end{multline*}
	Hence, $\cJ_2$ can be put in the slightly more compact form
	\begin{multline*}
		\cJ_2(\bfD;\bfeps,\bfeta) = \sum_{s = 0}^t\sum_{u = 0}^1(-1)^{s + u}\sum_{\substack{\vert\bfalpha\vert = t - s\\ \bfalpha\subset\bfeps}}\sum_{\substack{\vert\bfalpha'\vert,\vert\bfalpha''\vert=s\\\bfalpha',\bfalpha''\subset\bfeps\\ \bfalpha'\cap\bfalpha'' = \emptyset}}\sum_{\substack{\vert\bfdelta\vert = 1 - u\\ \bfdelta\subset\bfeta}}\sum_{\substack{\vert\bfdelta'\vert,\vert\bfdelta''\vert=u\\\bfdelta',\bfdelta''\subset\bfeta\\ \bfdelta'\cap\bfdelta'' = \emptyset}}\\
		\Lambda(\bfD;\bfalpha'\cup\bfdelta'',\bfalpha''\cup\bfdelta';\bfdelta\cup\bfeps\setminus(\bfalpha\cup\bfalpha'\cup\bfalpha''),\bfalpha\cup\bfeta\setminus(\bfdelta\cup\bfdelta'\cup\bfdelta'')).
	\end{multline*}
	Summing over $\bfeps$ and $\bfeta$, and using Lemma~\ref{lemma:rewritting_sums_3}, this simplifies further into
	\begin{multline*}
		\sum_{\substack{\vert\bfeps\vert = r + t - 1\\ \vert\bfeta\vert = r - t + 1}}\cJ_2(\bfD;\bfeps,\bfeta) = \sum_{s = 0}^t\sum_{u = 0}^1(-1)^{s + u}\sum_{\substack{\vert\bfeps\vert = r - s - u\\ \vert\bfeta\vert = r - s - u}}\sum_{\vert\bfalpha\vert = t - s}\sum_{\vert\bfalpha'\vert,\vert\bfalpha''\vert = s}\sum_{\vert\bfdelta\vert = 1 - u}\sum_{\vert\bfdelta'\vert,\vert\bfdelta''\vert = u}\\
		\Lambda(\bfD;\bfalpha'\cup\bfdelta'',\bfalpha''\cup\bfdelta';\bfdelta\cup\bfeps,\bfalpha\cup\bfeta).
	\end{multline*}
	Lastly, we use Lemma~\ref{lemma:rewritting_sums_3} with the pairs $(\bfalpha',\bfdelta'')$, $(\bfalpha'',\bfdelta')$, $(\bfdelta,\bfeps)$ and $(\bfalpha,\bfeta)$ to obtain
	\begin{multline}
		\label{eq:k_prdm_main_estimate_proof_even_case_second_term_second_part}
		\sum_{\substack{\vert\bfeps\vert = r + t - 1\\ \vert\bfeta\vert = r - t + 1}}\cJ_2(\bfD;\bfeps,\bfeta) = \sum_{s = 0}^t\sum_{u = 0}^1(-1)^{s + u}{r - s - u \choose t - s}{r - s - u \choose 1 - u}{s + u \choose u}^2\\
		\times\sum_{\vert\bfeps\vert,\vert\bfeta\vert = r - s - u}\sum_{\vert\bfalpha\vert,\vert\bfbeta\vert = s + u}\Lambda(\bfD;\bfalpha,\bfbeta;\bfeps,\bfeta).
	\end{multline}
	
	Combining \eqref{eq:k_prdm_main_estimate_proof_even_case_intermediate_summary}--\eqref{eq:k_prdm_main_estimate_proof_even_case_second_term_second_part} finishes the proof of Lemma~\ref{lemma:main_estimate_even_case}.
\end{proof}

\begin{proof}[Proof of Proposition~\ref{prop:cancellation_main_order} in the even case]
	The main ingredient of the proof is Lemma~\ref{lemma:main_estimate_even_case}. The first term in the right-hand side of \eqref{eq:main_estimate_even_case} is already of the desired form, so we do not change it. Regarding the remaining two terms, it is important to notice that, for a given $s$, the combinatorial factor in the second term is of order $N^{t - s - 1}$, whereas the one in the third term is of order $N^{t - s + 1}$. Hence, taking $\tau = N$ to ensure that both terms are of the same order, we obtain
	\begin{align*}
		\sum_{\vert\bfeps\vert,\vert\bfeta\vert = r - t}\sum_{\vert\bfalpha\vert,\vert\bfbeta\vert = t}\Lambda(\bfD;\bfalpha,\bfbeta;\bfeps,\bfeta) &\leq \sum_{s = 0}^{t - 1}C_{s,t}(r,N)\sum_{\vert\bfeps\vert,\vert\bfeta\vert = r - s}\sum_{\vert\bfalpha\vert,\vert\bfbeta\vert = s}\Lambda(\bfD;\bfalpha,\bfbeta;\bfeps,\bfeta)\\
		&\phantom{\leq} + \dfrac{r - t}{tN}\sum_{\vert\bfeps\vert,\vert\bfeta\vert = r - t}\sum_{\vert\bfalpha\vert,\vert\bfbeta\vert = t}\Lambda(\bfD;\bfalpha,\bfbeta;\bfeps,\bfeta)\\
		&\phantom{\leq} - \dfrac{1}{2tN}\sum_{\vert\bfeps\vert,\vert\bfeta\vert = r - t - 1}\sum_{\vert\bfalpha\vert,\vert\bfbeta\vert = t + 1}\Lambda(\bfD;\bfalpha,\bfbeta;\bfeps,\bfeta),
	\end{align*}
	for some $C_{s,t}$'s satisfying \eqref{eq:cancellation_main_order_condition_coefficients}. To conclude the proof of Proposition~\ref{prop:cancellation_main_order}, we still need to get rid of the last two terms. For the last term, we sum over $r$ and $\bfD$ and use that
	\begin{multline*}
		- \sum_{r = 0}^N\sum_{\vert\bfD\vert = N - r}\sum_{\vert\bfeps\vert,\vert\bfeta\vert = r - t - 1}\sum_{\vert\bfalpha\vert,\vert\bfbeta\vert = t + 1}\Lambda(\bfD;\bfalpha,\bfbeta;\bfeps,\bfeta)\\
		= -\sum_{\substack{\vert\bfalpha\vert,\vert\bfbeta\vert = t + 1\\ \bfalpha\cap\bfbeta = \emptyset}}\bigg(\sum_{\vert\bfA\vert - N - t - 1}\sgn(\bfalpha\cup\bfbeta,\bfA)c_{\bfA\cup\bfalpha}c_{\bfA\cup\bfbeta}\bigg)^2 \leq 0,
	\end{multline*}
	which results from Lemma~\ref{lemma:rewritting_sums}. To deal with the second term, we add the nonnegative quantity
	\begin{multline*}
		\dfrac{N - r}{tN}\sum_{\vert\bfD\vert = N - r}\sum_{\vert\bfeps\vert,\vert\bfeta\vert = r - t}\bigg(\sum_{\vert\bfalpha\vert = t}\sgn(\bfalpha,\bfeps\cup\bfeta)c_{\bfA\cup\bfalpha\cup\bfeps}c_{\bfA\cup\bfalpha\cup\bfeta}\bigg)^2\\
		= \dfrac{N - r}{tN}\sum_{s = 0}^{t}{N - r + t - s \choose t - s}\sum_{\vert\bfD\vert = N - r + t - s}\sum_{\vert\bfeps\vert,\vert\bfeta\vert = r - t}\sum_{\vert\bfalpha\vert,\vert\bfbeta\vert = s}\Lambda(\bfD;\bfalpha,\bfbeta;\bfeps,\bfeta),
	\end{multline*}
	which is a consequence of Lemmas~\ref{lemma:rewritting_sums}~and~\ref{lemma:rewritting_sums_3}. This allows us to write
	\begin{multline*}
		\sum_{r = 0}^N\sum_{\vert\bfD\vert = N - r}\sum_{\vert\bfeps\vert,\vert\bfeta\vert = r - t}\sum_{\vert\bfalpha\vert,\vert\bfbeta\vert = t}\Lambda(\bfD;\bfalpha,\bfbeta;\bfeps,\bfeta)\\
		\begin{aligned}[t]
			&\leq \sum_{s = 0}^{t - 1}\sum_{r = 0}^NC_{s,t}(r,N)\sum_{\vert\bfD\vert = N - r}\sum_{\vert\bfeps\vert,\vert\bfeta\vert = r - s}\sum_{\vert\bfalpha\vert,\vert\bfbeta\vert = s}\Lambda(\bfD;\bfalpha,\bfbeta;\bfeps,\bfeta)\\
			&\phantom{\leq} + \dfrac{N - t}{tN}\sum_{r = 0}^N\sum_{\vert\bfD\vert = N - r}\sum_{\vert\bfeps\vert,\vert\bfeta\vert = r - t}\sum_{\vert\bfalpha\vert,\vert\bfbeta\vert = t}\Lambda(\bfD;\bfalpha,\bfbeta;\bfeps,\bfeta),
		\end{aligned}
	\end{multline*}
	for some new coefficients $C_{s,t}$ that still satisfy \eqref{eq:cancellation_main_order_condition_coefficients}. Finally, shifting the last term to the left of the equation concludes the proof of Proposition~\ref{prop:cancellation_main_order} for $t$ even.
\end{proof}

\section{Conclusion of the proof of Theorem~\ref{th:HS_estimates}}

\begin{proof}[Proof of Theorem~\ref{th:HS_estimates}]
	We use a proof by induction to show that, for any $t\in\{0,\dots,k\}$, there exist a family of real coefficients
	\begin{equation*}
		\left\{C_{s,t}(r,N): s\in\{0,\dots,t\},0\leq r\leq N\right\}
	\end{equation*}
	and a nonnegative constant $C_t$ depending only on $t$ and $k$, and satisfying
	\begin{equation}
		\label{eq:cancellation_main_order_condition_coefficients2}
		\left\vert C_{s,t}(r,N)\right\vert \leq C_tN^{k - s},
	\end{equation}
	for all $0\leq r \leq N$, and such that, for all normalised $\Psi\in\mfH^{\wedge N}$ expanded as in Lemma~\ref{lemma:rewriting_hs_norm}, the following estimate holds:
	\begin{equation}
		\label{eq:main_theorem_induction_property}
		\Vert\Gamma^{(k)}\Vert_\HS^2 \leq \sum_{s = 0}^t\sum_{r = 0}^NC_{s,t}(r,N)\sum_{\vert\bfD\vert = N - r}\sum_{\vert\bfeps\vert,\vert\bfeta\vert = r - s}\sum_{\vert\bfalpha\vert,\vert\bfbeta\vert = s}\Lambda(\bfD;\bfalpha,\bfbeta;\bfeps,\bfeta).
	\end{equation}
	The induction is made over decreasing values of $t$. Thanks to Lemma~\ref{lemma:rewriting_hs_norm}, it is easy to see that \eqref{eq:main_theorem_induction_property} is true for $t = k$. Suppose now that \eqref{eq:main_theorem_induction_property} is true for some $t \geq 1$, and let us prove it for $t - 1$. To do so, we would like to apply Proposition~\ref{prop:cancellation_main_order} to get rid of the $s = t$ term in \eqref{eq:cancellation_main_order_condition_coefficients2}. However, we cannot do so directly because it is not of the right form: there is an extra factor $C_{s,t}(r,N)$ in the sum over $r$. To circumvent this, we add the nonnegative quantity
	\begin{multline}
		\label{eq:main_theorem_large_nonnegative_quantity}
		\sum_{r=0}^N\left[C_tN^{k - t} - C_{t,t}(r,N)\right]\sum_{\vert\bfD\vert = N - r}\sum_{\substack{\vert\bfeps\vert,\vert\bfeta\vert = r - t\\ \bfeps\cap\bfeta = \emptyset}}\Bigg(\sum_{\vert\bfalpha\vert = t}\sgn(\bfalpha,\bfeps\cup\bfeta)c_{\bfD\cup\bfalpha\cup\bfeps}c_{\bfD\cup\bfalpha\cup\bfeta}\Bigg)^2\\
		\begin{multlined}[t]
			= \sum_{r=0}^N\left[C_tN^{k - t} - C_{t,t}(r,N)\right]\sum_{s = 0}^{t}{N - r + t - s \choose t - s}\\
			\times\sum_{\vert\bfD\vert = N - r + t - s}\sum_{\vert\bfeps\vert,\vert\bfeta\vert = r}\sum_{\vert\bfalpha\vert,\vert\bfbeta\vert = s}\Lambda(\bfD;\bfalpha,\bfbeta;\bfeps,\bfeta),
		\end{multlined}
	\end{multline}
	which we rewrote using Lemmas~\ref{lemma:rewritting_sums}~and~\ref{lemma:rewritting_sums_3}. Doing so and using the estimate
	\begin{equation*}
		\left[C_tN^{k - t} - C_{t,t}(r,N)\right]{N - r + t - s \choose t - s} \leq CN^{k - s},
	\end{equation*}
	we can find new coefficients $C_{s,t-1}$ and a nonnegative constant $C_{t - 1}$ satisfying
	\begin{equation}
		\label{eq:cancellation_main_order_condition_coefficients3}
		\left\vert C_{s,t-1}(r,N)\right\vert \leq C_{t - 1}N^{k - s},
	\end{equation}
	and such that
	\begin{align*}
		\Vert\Gamma^{(k)}\Vert_\HS^2 &\leq C_tN^{k - t}\sum_{r = 0}^N\sum_{\vert\bfD\vert = N - r}\sum_{\vert\bfeps\vert,\vert\bfeta\vert = r - t}\sum_{\vert\bfalpha\vert,\vert\bfbeta\vert = t}\Lambda(\bfD;\bfalpha,\bfbeta;\bfeps,\bfeta)\\
		&\phantom{\leq} + \sum_{s = 0}^{t - 1}\sum_{r = 0}^NC_{s,t-1}(r,N)\sum_{\vert\bfD\vert = N - r}\sum_{\vert\bfeps\vert,\vert\bfeta\vert = r - s}\sum_{\vert\bfalpha\vert,\vert\bfbeta\vert = s}\Lambda(\bfD;\bfalpha,\bfbeta;\bfeps,\bfeta).
	\end{align*}
	Applying Proposition~\ref{prop:cancellation_main_order} to the first term and changing again the $C_{s,t-1}$ coefficients and the constant $C_{t-1}$ in \eqref{eq:cancellation_main_order_condition_coefficients3}, we obtain
	\begin{equation*}
		\Vert\Gamma^{(k)}\Vert_\HS^2 \leq \sum_{s = 0}^{t - 1}\sum_{r = 0}^NC_{s,t-1}(r,N)\sum_{\vert\bfD\vert = N - r}\sum_{\vert\bfeps\vert,\vert\bfeta\vert = r - s}\sum_{\vert\bfalpha\vert,\vert\bfbeta\vert = s}\Lambda(\bfD;\bfalpha,\bfbeta;\bfeps,\bfeta),
	\end{equation*}
	which concludes the induction step.
	
	Finally, we use \eqref{eq:main_theorem_induction_property} with $t = 0$ and the normalisation of $\Psi$ to deduce
	\begin{equation*}
		\Vert\Gamma^{(k)}\Vert_\HS^2 \leq C_kN^k\sum_{r = 0}^N\sum_{\vert\bfD\vert = N - r}\sum_{\vert\bfeps\vert,\vert\bfeta\vert = r}\Lambda(\bfD;\bfeps,\bfeta) = C_kN^k.
	\end{equation*}
	This concludes the proof of Theorem~\ref{th:HS_estimates}.
\end{proof}

	\noindent
	\textbf{Acknowledgments.} The author would like to express his sincere gratitude to Phan Thành Nam and Arnaud Triay for the fruitful discussions and their continued support. Partial support by the Deutsche Forschungsgemeinschaft (DFG, German Research Foundation) through the TRR 352 Project ID. 470903074 and by the European Research Council through the ERC CoG RAMBAS  Project Nr. 101044249 is acknowledged.
	
	\printbibliography

	\appendix
	\label{appendix:proof_combinatorial_results}
	\section{Proof of combinatorial results}
	\begin{proof}[Proof of Lemma~\ref{lemma:rewritting_sums}]
	Due to the condition \eqref{eq:rewritting_sums_index_distinct_condition}, the multi-indices $\bfD,\bfA$ and $\bfB$ in the right-hand side of \eqref{eq:rewritting_sums} satisfy the condition $\bfD\cap\bfA = \bfD\cap\bfB = \emptyset$. Consider the sets
	\begin{equation*}
		S_1 \coloneqq \left\{(\bfA,\bfB): \vert \bfA\vert = N, \vert\bfB\vert = M, \textmd{ $\bfA$ and $\bfB$ disjoint}\right\}
	\end{equation*}
	and
	\begin{equation*}
		S_2 \coloneqq \left\{(\bfD,\bfA,\bfB): \vert\bfD\cup\bfA\vert = N, \vert\bfD\cup\bfB\vert = M, \textmd{ $\bfA, \bfB$ and $\bfD$ pairwise disjoint}\right\}.
	\end{equation*}
	Consider the map $h: S_1\rightarrow S_2$ defined by
	\begin{equation*}
		h: (\bfA,\bfB) \mapsto (\bfA\cap \bfB,\bfA\setminus \bfB,\bfB \setminus \bfA).
	\end{equation*}
	It is easy to see that $h$ is a bijective map from $S_1$ to $S_2$. Hence, we may write
	\begin{equation*}
		\sum_{(\bfA,\bfB)\in S_1}f(\bfA,\bfB) = \sum_{(\bfD,\bfA,\bfB)\in S_2}f(h^{-1}(\bfD,\bfA,\bfB)) = \sum_{(\bfD,\bfA,\bfB)\in S_2}f(\bfD\cup\bfA,\bfD\cup\bfB).
	\end{equation*}
	Parametrising the length of $\bfA$ by $r$ concludes the proof of Lemma~\ref{lemma:rewritting_sums}.
\end{proof}

\begin{proof}[Proof of Lemma~\ref{lemma:rewritting_sums_2}]
	Consider the sets
	\begin{equation*}
		S_1 \ceqq \left\{(\bfA,\bfB): \vert\bfA\vert = N, \vert\bfB\vert = M, \textmd{ $\bfA$ and $\bfB$ disjoint}\right\}
	\end{equation*}
	and
	\begin{equation*}
		S_2 \ceqq \left\{(\bfA,\bfB): \vert\bfA\vert = N + M, \vert\bfB\vert = M, \bfB\subset\bfA\right\},
	\end{equation*}
	as well as the bijective map $h:S_1 \rightarrow S_2$ defined by
	\begin{equation*}
		h(\bfA,\bfB) \mapsto (\bfA\cup\bfB,\bfB).
	\end{equation*}
	The equality \eqref{eq:rewritting_sums_3} follows from the fact that, for all $(\bfA,\bfB)\in S_2$,
	\begin{equation*}
		h^{-1}(\bfA,\bfB) = (\bfA\setminus\bfB,\bfB).
	\end{equation*}
\end{proof}

\begin{proof}[Proof of Lemma~\ref{lemma:rewritting_sums_3}]
	Applying Lemma~\ref{lemma:rewritting_sums_3}, we deduce
	\begin{equation*}
		\sum_{\vert\bfA\vert = N}\sum_{\vert\bfB\vert = M}f(\bfA\cup\bfB) = \sum_{\vert\bfA\vert = N + M}\sum_{\substack{\vert\bfB\vert = M\\ \bfB\subset \bfA}}f(\bfA) = {N + M \choose N}\sum_{\vert\bfA\vert = N + M}f(\bfA).
	\end{equation*}
\end{proof}

\end{document}